\documentclass[a4paper,UKenglish,cleveref, autoref, thm-restate, pdfa]{lipics-v2021}

\newcommand{\N}{\mathbb{N}}

\newcommand{\B}{\{0,1\}}

\newcommand{\Plays}{\mathsf{Plays}}
\newcommand{\Hist}{\mathsf{Hist}}
\newcommand{\occ}[1]{\mathsf{Occ}({#1})}
\newcommand{\infOcc}[1]{\mathsf{Inf}({#1})}

\newcommand{\payoff}[1]{\mathsf{pay}({#1})}
\newcommand{\won}[1]{\mathsf{won}({#1})}
\newcommand{\payoffObj}[2]{\mathsf{pay}_{#1}({#2})}
\newcommand{\update}[1]{\mathsf{upd}({#1})}

\newcommand{\strategyProfile}{\sigma}
\newcommand{\outcome}[1]{\mathsf{out}({#1})}

\newcommand{\Playsigma}[1]{\mathsf{Plays}_{#1}}
\newcommand{\Histsigma}[1]{\mathsf{Hist}_{#1}}
\newcommand{\Playsigmazero}{\Playsigma{\sigma_0}}

\newcommand{\nbrObjectives}{t}
\newcommand{\Obj}{\Omega}
\newcommand{\ObjPlayer}[1]{\Omega_{#1}}
\newcommand{\target}{T} 
\newcommand{\reach}[1]{\mathsf{Reach}(#1)}

\newcommand{\parity}[1]{\mathsf{Parity}(#1)}

\newcommand{\Buchi}[1]{\mathsf{B\ddot{u}chi}(#1)}

\newcommand{\BooleanBuchi}[1]{\mathsf{BooleanB\ddot{u}chi}(#1)}

\newcommand{\p}{$\mathsf{P}$}
\newcommand{\np}{$\mathsf{NP}$}
\newcommand{\npHard}{$\mathsf{NP}$-hard}
\newcommand{\npComplete}{$\mathsf{NP}$-complete}

\newcommand{\pspaceHard}{$\mathsf{PSPACE}$-hard}
\newcommand{\pspaceComplete}{$\mathsf{PSPACE}$-complete}
\newcommand{\nexptime}{$\mathsf{NEXPTIME}$}
\newcommand{\nexptimeHard}{$\mathsf{NEXPTIME}$-hard}
\newcommand{\nexptimeComplete}{$\mathsf{NEXPTIME}$-complete}

\newcommand{\FPT}{$\mathsf{FPT}$}

\newcommand{\problemParam}{k}
\newcommand{\problem}{Stackelberg-Pareto Synthesis problem}
\newcommand{\problemAb}{SPS problem}
\newcommand{\setCover}{Set Cover problem}
\newcommand{\setCoverAb}{SC problem}
\newcommand{\succinctSetCover}{Succinct Set Cover problem}
\newcommand{\succinctSetCoverAb}{SSC problem}
\newcommand{\dominatingSet}{Succinct Dominating Set problem}
\newcommand{\dominatingSetAb}{SDS problem}

\newcommand{\challengerProver}{Challenger-Prover}
\newcommand{\challengerProverAb}{C-P}
\newcommand{\prov}{\mathcal{P}}
\newcommand{\chal}{\mathcal{C}}
\newcommand{\provWit}{W}

\newcommand{\fixed}{$\sigma_0$-fixed}
\newcommand{\paretoOptimal}{\fixed{} Pareto-optimal}
\newcommand{\fixedStrategy}[1]{{#1}-fixed}
\newcommand{\paretoOptimalStrategy}[1]{\fixedStrategy{#1} Pareto-optimal}
\newcommand{\game}{Stackelberg-Pareto game}
\newcommand{\games}{Stackelberg-Pareto games}
\newcommand{\gameAb}{SP game}
\newcommand{\gamesAb}{SP games}
\newcommand{\paretoSet}[1]{P_{#1}}
\newcommand{\paretoSetSize}{|\Wit{\sigma_0}|}
\newcommand{\dominatedPlays}[1]{\Omega^{<}(\paretoSet{#1})}
\newcommand{\deviations}[1]{\mathsf{Dev}(\Wit{#1})}

\newcommand{\punStrat}[1]{\sigma^{\mathsf{Pun}}_{#1}}

\newcommand{\Wit}[1]{\mathsf{Wit}_{#1}}

\newcommand{\prefStrat}[1]{\mathsf{Wit}_{\sigma_0}(#1)}
\newcommand{\W}{W}

\newcommand{\region}{region}

\newcommand{\sect}{section}

\newcommand{\internal}{internal}
\newcommand{\final}{terminal}
\newcommand{\Reg}[1]{\mathsf{Reg}({#1})}

\newcommand{\Witproperty}{region-tree structure}

\usepackage{stmaryrd} 
\usepackage{tikz}
\title{Stackelberg-Pareto Synthesis (full version)}

\author{Véronique Bruyère}{Université de Mons (UMONS), Belgium}{}{}{}

\author{Jean-François Raskin}{Université libre de Bruxelles (ULB), Belgium}{}{}{}

\author{Clément Tamines}{Université de Mons (UMONS), Belgium}{}{}{}

\authorrunning{V. Bruyère, J.-F. Raskin and C. Tamines}

\Copyright{Véronique Bruyère, Jean-François Raskin and Clément Tamines} 


\category{} 

\relatedversion{} 

\funding{This work is partially supported by the PDR project Subgame perfection in graph games (F.R.S.-FNRS), the ARC project Non-Zero Sum Game Graphs: Applications to Reactive Synthesis and Beyond (Fédération Wallonie-Bruxelles), the EOS project Verifying Learning Artificial Intelligence Systems (F.R.S.-FNRS and FWO), and the COST Action 16228 GAMENET (European Cooperation in Science and Technology).}

\acknowledgements{}

\nolinenumbers 

\hideLIPIcs  

\EventEditors{John Q. Open and Joan R. Access}
\EventNoEds{2}
\EventLongTitle{42nd Conference on Very Important Topics (CVIT 2016)}
\EventShortTitle{CVIT 2016}
\EventAcronym{CVIT}
\EventYear{2016}
\EventDate{December 24--27, 2016}
\EventLocation{Little Whinging, United Kingdom}
\EventLogo{}
\SeriesVolume{42}
\ArticleNo{23}

\begin{document}

\maketitle

\begin{abstract}
In this paper, we study the framework of two-player Stackelberg games played on graphs in which Player~$0$ announces a strategy and Player~$1$ responds rationally with a strategy that is an optimal response. While it is usually assumed that Player~$1$ has a single objective, we consider here the new setting where he has several. In this context, after responding with his strategy, Player~$1$ gets a payoff in the form of a vector of Booleans corresponding to his satisfied objectives. Rationality of Player~$1$ is encoded by the fact that his response must produce a Pareto-optimal payoff given the strategy of Player~$0$. We study the \problem{} which asks whether Player~$0$ can announce a strategy which satisfies his objective, whatever the rational response of Player~$1$. For games in which objectives are either all parity or all reachability objectives, we show that this problem is fixed-parameter tractable and \nexptimeComplete{}. This problem is already \npComplete{} in the simple case of reachability objectives and graphs that are trees. 
\end{abstract}

\section{Introduction} 
Two-player zero-sum infinite-duration games played on graphs are a mathematical model used to formalize several important problems in computer science, such as \emph{reactive system synthesis}.  In this context, see e.g.~\cite{PnueliR89}, the graph represents the possible interactions between the system and the environment in which it operates. One player models the system to synthesize, and the other player models the (uncontrollable) environment. In this classical setting, the objectives of the two players are opposite, that is, the environment is \emph{adversarial}. Modelling the environment as fully adversarial is usually a \emph{bold abstraction} of reality as it can be composed of one or several components, each of them having their own objective. 

In this paper, we consider the framework of \emph{Stackelberg games}~\cite{Stackelberg37}, a richer non-zero-sum setting, in which Player~$0$ (the system) called \emph{leader} announces his strategy and then Player~$1$ (the environment) called \emph{follower} plays rationally by using a strategy that is an optimal response to the leader's strategy. This framework captures the fact that in practical applications, a strategy for interacting with the environment is committed before the interaction actually happens. The goal of the leader is to announce a strategy that guarantees him a payoff at least equal to some given threshold. In the specific case of Boolean objectives, the leader wants to see his objective being satisfied. The concept of leader and follower is also present in the framework of \emph{rational synthesis}~\cite{FismanKL10,KupfermanPV16} with the difference that this framework considers several followers, each of them with their own Boolean objective. In that case, rationality of the followers is modeled by assuming that the environment settles to an equilibrium (e.g. a Nash equilibrium) where each component (composing the environment) is considered to be an \emph{independent selfish individual}, excluding cooperation scenarios between components or the possibility of coordinated rational multiple deviations. Our work proposes a novel and natural \emph{alternative} in which the single follower, modeling the environment, has several objectives that he wants to satisfy. After responding to the leader with his own strategy, Player~$1$ receives a vector of Booleans which is his payoff in the corresponding outcome. Rationality of Player~$1$ is encoded by the fact that he only responds in such a way to receive \emph{Pareto-optimal payoffs}, given the strategy announced by the leader. This setting  encompasses scenarios where, for instance, several components can collaborate and agree on trade-offs. The goal of the leader is therefore to announce a strategy that guarantees him to satisfy his own objective, whatever the response of the follower which ensures him a Pareto-optimal payoff. The problem of deciding whether the leader has such a strategy is called the \emph{Stackelberg-Pareto Synthesis problem} (\problemAb{}). 

\subparagraph*{Contributions.} In addition to the definition of the new setting, our main contributions are the following ones. We consider the general class of $\omega$-regular objectives modelled by \emph{parity} conditions and also consider the case of \emph{reachability} objectives for their simplicity\footnote{Indeed, in the classical context of two-player zero-sum games, solving reachability games is in \p{} whereas solving parity games is only known to be in $\mathsf{NP} \cap \text{co-}\mathsf{NP}$, see e.g.~\cite{2001automata}.}.
We provide a thorough analysis of the complexity of solving the \problemAb{} for both objectives. Our results are interesting and singular both from a theoretical and practical point of view. 

First, we show that the \problemAb{} is \emph{fixed-parameter tractable} (\FPT) for reachability objectives when the number of objectives of the follower is a parameter (\autoref{thm:FPT-reach}) and for parity objectives when, in addition, the maximal priority used in each priority function is also a parameter of the complexity analysis (\autoref{thm:FPT-parity}). These are important results as it is expected that, in practice, the \emph{number} of objectives of the environment is limited to a few. To obtain these results, we develop a reduction from our non-zero-sum games to a zero-sum game in which the protagonist, called \emph{Prover}, tries to show the existence of a solution to the problem, while the antagonist, called \emph{Challenger}, tries to disprove it. This zero-sum game is defined in a \emph{generic way}, independently of the actual objectives used in the initial game, and can then be easily adapted according to the case of reachability or parity objectives.

Second, we prove that the \problemAb{} is \nexptimeComplete{} for both reachability and parity objectives (Theorems~\ref{thm:nexptime}, \ref{thm:nexptimehard-reach} and \ref{thm:nexptimehard-parity}), and that it is already \npComplete{} in the simple setting of reachability objectives and graphs that are trees (\autoref{thm:npcomplete}). 
To the best of our knowledge, this is the first \nexptime-completeness result for a natural class of games played on graphs. To obtain the hardness for \nexptime{}, we present a natural \emph{succinct version} of the set cover problem that is complete for this class (\autoref{thm:ssc-completeness}), a result of potential independent interest. We then show how to reduce this problem to the \problemAb{}. To obtain the \nexptime{}-membership of the \problemAb{}, we have shown that exponential-size solutions exist for positive instances of the \problemAb{} and this allows us to design a nondeterministic exponential-time algorithm. Unfortunately, it was not possible to use the \FPT{} algorithm mentioned above to show this membership due to its too high time complexity; conversely, our \nexptime{} algorithm is not \FPT{}.

\subparagraph*{Related Work.} Rational synthesis is introduced in~\cite{FismanKL10} for $\omega$-regular objectives in a setting where the followers are cooperative with the leader, and later in~\cite{KupfermanPV16} where they are adversarial. Precise complexity results for various $\omega$-regular objectives are established in~\cite{ConduracheFGR16} for both settings. Those complexities differ from the ones of the problem studied in this paper. Indeed, for reachability objectives, adversarial rational synthesis is \pspaceComplete, while for parity objectives, its precise complexity is not settled (the problem is \pspaceHard{} and in \nexptime{}). Extension to non-Boolean payoffs, like mean-payoff or discounted sum, is studied in~\cite{GuptaS14,GuptaS14c} in the cooperative setting and in~\cite{BalachanderGR20,FiliotGR20} in the adversarial setting. 

When several players (like the followers) play with the aim to satisfy their objectives, several solution concepts exist such as Nash equilibrium~\cite{Nas50}, subgame perfect equilibrium~\cite{selten}, secure equilibria~\cite{DBLP:conf/tacas/ChatterjeeH07,DBLP:journals/tcs/ChatterjeeHJ06}, or admissibility~\cite{Berwanger07,BrenguierRS15}. The constrained existence problem, close to the cooperative rational synthesis problem, is to decide whether there exists a solution concept such that the payoff obtained by each player is larger than some threshold. Let us mention~\cite{ConduracheFGR16,Ummels08,UmmelsW11} for results on the constrained existence for Nash equilibria and~\cite{Raskin2021,BrihayeBGRB20,Ummels06} for such results for subgame perfect equilibria. Rational verification is studied in~\cite{GutierrezNPW19,GutierrezNPW20}. This problem (which is not a synthesis problem) is to decide whether a given LTL formula is satisfied by the outcome of all Nash equilibria (resp. some Nash equilibrium). The interested reader can find more pointers to works on non-zero-sum games for reactive synthesis in~\cite{DBLP:conf/lata/BrenguierCHPRRS16,Bruyere17}.

\subparagraph*{Structure.} The paper is structured as follows. In Section~\ref{sec:prelim}, we introduce the class of \games{} and the \problemAb{}. We show in Section~\ref{sec:newfpt} that the \problemAb{} is in \FPT{} for 
reachability and parity objectives. The complexity class of this problem is studied in Section~\ref{sec:nexptimecomplete} where we prove that it is \nexptimeComplete{} and \npComplete{} in case of reachability objectives and graphs that are trees. In Section \ref{sec:conclusion}, we provide a conclusion and discuss future work.

\section{Preliminaries and Stackelberg-Pareto Synthesis Problem} \label{sec:prelim}

This section introduces the class of two-player \games{} in which the first player has a single objective and the second has several. We present a decision problem on those games called the \problem{}, which we study in this paper.

\subsection{Preliminaries}

\subparagraph*{Game Arena.}
A \emph{game arena} is a tuple \sloppy $G = (V, V_0, V_1, E, v_0)$ where $(V,E)$ is a finite directed graph such that: \emph{(i)} $V$ is the set of vertices and $(V_0, V_1)$ forms a partition of $V$ where $V_0$ (resp.\ $V_1$) is the set of vertices controlled by Player~$0$ (resp.\ Player~$1$), \emph{(ii)} $E \subseteq V \times V$ is the set of edges such that each vertex $v$ has at least one successor $v'$, i.e., $(v,v') \in E$, and \emph{(iii)} $v_0 \in V$ is the initial vertex. We call a game arena a \emph{tree arena} if it is a tree in which every leaf vertex has itself as its only successor. A \emph{sub-arena} $G'$ with a set $V' \subseteq V$ of vertices and initial vertex $v'_0 \in V'$ is a game arena defined from $G$ as expected.

\subparagraph*{Plays.}
A \emph{play} in a game arena $G$ is an infinite sequence of vertices $\rho = v_0 v_1 \ldots \in V^{\omega}$ such that it starts with the initial vertex $v_0$ and $(v_j,v_{j+1}) \in E$ for all $j \in \N$. \emph{Histories} in $G$ are finite sequences $h = v_0 \ldots v_j \in V^+$ defined similarly. A history is \emph{elementary} if it contains no cycles. We denote by $\Plays_G$ the set of plays in $G$. We write $\Hist_G$ (resp.\ $\Hist_{G,i}$) the set of histories (resp.\ histories ending with a vertex in $V_i$). We use the notations $\Plays$, $\Hist$, and $\Hist_i$ when $G$ is clear from the context. We write $\occ{\rho}$ the set of vertices occurring in $\rho$ and $\infOcc{\rho}$ the set of vertices occurring infinitely often in $\rho$. 

\subparagraph*{Strategies.}
A \emph{strategy} $\sigma_i$ for Player~$i$ is a function $\sigma_i\colon \Hist_i \rightarrow V$ assigning to each history $hv \in \Hist_i$ a vertex $v' = \sigma_i(hv)$ such that $(v,v') \in E$. It is \emph{memoryless} if $\sigma_i(hv) = \sigma_i(h'v)$ for all histories $hv, h'v$ ending with the same vertex $v \in V_i$. More generally, it is \emph{finite-memory} if it can be encoded by a Moore machine ${\cal M}$~\cite{2001automata}. The \emph{memory size} of $\sigma_i$ is the number of memory states of $\cal M$. In particular, $\sigma_i$ is memoryless when it has a memory size of one.

Given a strategy $\sigma_i$ of Player~$i$, a play $\rho = v_0 v_1 \ldots$ is \emph{consistent} with $\sigma_i$ if $v_{j+1} = \sigma_i(v_0 \ldots v_j)$ for all $j \in \N$ such that $v_j \in V_i$. Consistency is naturally extended to histories. We denote by $\Playsigma{\sigma_i}$ (resp.\ $\Histsigma{\sigma_i}$) the set of plays (resp.\ histories) consistent with $\sigma_i$. A \emph{strategy profile} is a tuple $\strategyProfile = (\sigma_0, \sigma_1)$ of strategies, one for each player. We write $\outcome{\strategyProfile}$ the unique play consistent with both strategies and we call it the \emph{outcome} of $\strategyProfile$. 

\subparagraph*{Objectives.}
An \emph{objective} for Player~$i$ is a set of plays $\Obj \subseteq \Plays$. A play $\rho$ \emph{satisfies} the objective $\Obj$ if $\rho \in \Obj$. In this paper, we focus on the two following $\omega$-regular objectives. Let $\target \subseteq V$ be a subset of vertices called a \emph{target set}, the \emph{reachability} objective $\reach{\target} = {\{\rho \in \Plays \mid \occ{\rho} \cap \target \neq \emptyset \}}$ asks to visit at least one vertex of $\target$. Let $c : V \rightarrow \mathbb{N}$ be a function called a \emph{priority function} which assigns an integer to each vertex in the arena, the \emph{parity} objective $\parity{c} = \{\rho \in \Plays \mid \min_{v \in \infOcc{\rho}}(c(v)) \text{ is even}\}$ asks that the minimum priority visited infinitely often be even.

\subsection{Stackelberg-Pareto Synthesis Problem}

\subparagraph*{Stackelberg-Pareto Games.} 
A \emph{\game{}} (\gameAb{}) $\mathcal{G} = (G, \ObjPlayer{0},\ObjPlayer{1}, \dots, \ObjPlayer{\nbrObjectives})$ is composed of a game arena $G$, an objective $\ObjPlayer{0}$ for Player~$0$ and $\nbrObjectives \geq 1$ objectives $\ObjPlayer{1}, \dots, \ObjPlayer{\nbrObjectives}$ for Player~$1$. In this paper, we focus on \gamesAb{} where the objectives are either all reachability or all parity objectives and call such games \emph{reachability} (resp.\ \emph{parity}) \emph{\gamesAb{}}.

\subparagraph*{Payoffs in SP Games.}
The \emph{payoff} of a play $\rho \in \Plays$ corresponds to the vector of Booleans $\payoff{\rho} \in \{0,1\}^{\nbrObjectives}$ such that for all $i \in \{1, \dots, \nbrObjectives\}$, $\payoffObj{i}{\rho} = 1$ if $\rho \in \ObjPlayer{i}$, and $\payoffObj{i}{\rho} = 0$ otherwise. Note that we omit to include Player~$0$ when discussing the payoff of a play. Instead we say that a play $\rho$ is \emph{won} by Player~$0$ if $\rho \in \ObjPlayer{0}$ and we write $\won{\rho} = 1$, otherwise it is \emph{lost} by Player~$0$ and we write $\won{\rho} = 0$. We write $(\won{\rho},\payoff{\rho})$ the \emph{extended payoff} of $\rho$. Given a strategy profile $\strategyProfile$, we write $\won{\strategyProfile} = \won{\outcome{\strategyProfile}}$ and $\payoff{\strategyProfile} = \payoff{\outcome{\strategyProfile}}$.  For reachability \gamesAb{}, since reachability objectives are prefix-dependant and given a history $h \in \Hist$, we also define $\won{h}$ and $\payoff{h}$ as done for plays.

We introduce the following partial order on payoffs. Given two payoffs $p = (p_1, \dots, p_\nbrObjectives)$ and $p' = (p'_1, \dots, p'_\nbrObjectives)$ such that $p, p' \in \{0,1\}^{\nbrObjectives}$, we say that $p'$ is \emph{larger} than $p$ and write $p  \leq p'$ if $p_i \leq p'_i$ for all $i \in \{1, \dots, \nbrObjectives\}$. Moreover, when it also  holds that $p_i < p'_i$ for some $i$, we say that $p'$ is \emph{strictly larger} than $p$ and we write $p < p'$. A subset of payoffs $P \subseteq \{0,1\}^\nbrObjectives$ is an \emph{antichain} if it is composed of pairwise incomparable payoffs with respect to $\leq$.

\subparagraph*{Stackelberg-Pareto Synthesis Problem.} Given a strategy $\sigma_0$ of Player~$0$, we consider the set of payoffs of plays consistent with $\sigma_0$ which are \emph{Pareto-optimal}, i.e., maximal with respect to $\leq$. We write this set $\paretoSet{\sigma_0} =  \max \{\payoff{\rho} \mid \rho \in \Playsigmazero \}$. Notice that it is an antichain. We say that those payoffs are \emph{\paretoOptimal} and write $|\paretoSet{\sigma_0}|$ the number of such payoffs. A play $\rho \in \Playsigmazero$ is called \paretoOptimal{} if its payoff $\payoff{\rho}$ is in $\paretoSet{\sigma_0}$.

The problem studied in this paper asks whether there exists a strategy $\sigma_0$ for Player~$0$ such that every play in $\Playsigmazero$ which is \paretoOptimal{} satisfies the objective of Player~$0$. This corresponds to the assumption that given a strategy of Player~$0$, Player~$1$ will play \emph{rationally}, that is, with a strategy $\sigma_1$ such that  $\outcome{(\sigma_0, \sigma_1)}$ is \paretoOptimal{}. 
It is therefore sound to ask that Player~$0$ wins against such rational strategies.

\begin{definition}
	Given an \gameAb{}, the \emph{\problem{}} (\problemAb{}) is to decide whether there exists a strategy $\sigma_0$ for Player~$0$ (called a \emph{solution}) such that for each strategy profile ${\strategyProfile} = {(\sigma_0, \sigma_1)}$ with $\payoff{\strategyProfile} \in \paretoSet{\sigma_0}$, it holds that $\won{\strategyProfile} = 1$.
\end{definition}

\subparagraph*{Witnesses.} Given a strategy $\sigma_0$ that is a solution to the \problemAb{} and any payoff $p \in \paretoSet{\sigma_0}$, for each play $\rho$ consistent with $\sigma_0$ such that $\payoff{\rho} = p$ it holds that $\won{\rho}=1$. For each $p \in \paretoSet{\sigma_0}$, we arbitrarily select such a play which we call a \emph{witness} (of $p$). We denote by $\Wit{\sigma_0}$ the set of all witnesses, of which there are as many as payoffs in $\paretoSet{\sigma_0}$. In the sequel, it is useful to see this set as a tree composed of $\paretoSetSize$ branches. Additionally for a given history $h \in \Hist$, we write $\prefStrat{h}$ the set of witnesses for which $h$ is a prefix, i.e., $\prefStrat{h} = \{ \rho \in \Wit{\sigma_0} \mid h$ is prefix of $\rho \}$. Notice that $\prefStrat{h} = \Wit{\sigma_0}$ when $h = v_0$ and that $\prefStrat{h}$ decreases as $h$ increases, until it contains a single value or becomes empty. 

\begin{example}
\label{ex:example}

\begin{figure}
	\centering
		\resizebox{0.7\textwidth}{!}{%
		\begin{tikzpicture}[->,>=stealth, shorten >=1pt,auto]
		
		\node[draw, rectangle, minimum size=0.8cm] (v0) at (-0.25,0){$v_0$};
		\node[draw, circle, minimum size=0.8cm] (v1) at (1.5,0.75){$v_1$};
		\node[draw, rectangle, minimum size=0.8cm] (x) at (1.5,-0.75){$v_2$};
		\node[draw, circle, minimum size=0.8cm] (v2) at (3.25,0){$v_3$};
		\node[draw, circle, minimum size=0.8cm] (v3) at (3.25,-1.5){$v_4$};
		\node[draw, rectangle, minimum size=0.8cm] (v4) at (4.75,0.75){$v_5$};			
		\node[draw, circle, minimum size=0.8cm] (v5) at (4.75,-0.75){$v_7$};
		\node[draw, circle, minimum size=0.8cm] (v6) at (6.25,0.75){$v_6$};
				
		\draw[-stealth] (v0) edge [] node {} (v1);
		\draw[-stealth] (v0) edge [] node {} (x);
		\draw[-stealth] (x) edge [] node {} (v2);
		\draw[-stealth] (x) edge [] node {} (v3);

		\draw[-stealth] (v1) edge [loop right] node[right] {\small $(0, (0,0,1))$} (v1);
		
		\draw[-stealth] (v3) edge [loop right] node[right] {\small $( 0, (1,0,0))$} (v3);

		\draw[-stealth] (v2) edge [] node {} (v4);
		\draw[-stealth] (v4) edge [bend left] node {} (v2);

		\draw[-stealth] (v2) edge [] node {} (v5);
		
		\draw[-stealth] (v5) edge [loop right] node {\small $(1, (1,1,0))$} (v5);
		
		\draw[-stealth] (v4) edge [] node {} (v6);
		
		\draw[-stealth] (v6) edge [loop right] node {\small$(1, (0,1,1))$} (v6);

		\end{tikzpicture}
		}%
	
	\caption{A reachability \gameAb{}.}
	\label{example_memory}
\end{figure}

Consider the reachability \gameAb{} with arena $G$ depicted in Figure \ref{example_memory} in which Player~$1$ has $\nbrObjectives = 3$ objectives. The vertices of Player~$0$ (resp.\ Player~$1$) are depicted as ellipses (resp.\ rectangles)\footnote{This convention is used throughout this paper.}. Every objective in the game is a reachability objective defined as follows: $\ObjPlayer{0} = \reach{\{v_6, v_7\}}$, $\ObjPlayer{1} = \reach{\{v_4, v_7\}}$, $\ObjPlayer{2} = \reach{\{v_3\}}$, $\ObjPlayer{3} = \reach{\{v_1, v_6\}}$. The extended payoff of plays reaching vertices from which they can only loop is displayed in the arena next to those vertices, and the extended payoff of play $v_0 v_2 (v_3v_5)^\omega$ is $(0, (0,1,0))$.

Consider the memoryless strategy $\sigma_0$ of Player~$0$ such that he chooses to always move to $v_5$ from $v_3$. The set of payoffs of plays consistent with $\sigma_0$ is $\{(0,0,1), (0,1,0), (1, 0, 0), (0, 1, 1)\}$ and the set of those that are Pareto-optimal is $\paretoSet{\sigma_0} = \{(1, 0, 0), (0, 1, 1)\}$. Notice that play $\rho = v_0 v_2 (v_4)^\omega$ is consistent with $\sigma_0$, has payoff $(1, 0, 0)$ and is lost by Player~$0$. Strategy $\sigma_0$ is therefore not a solution to the \problemAb{}. In this game, there is only one other memoryless strategy for Player~$0$, where he chooses to always move to $v_7$ from $v_3$. One can verify that it is again not a solution to the \problemAb{}. 

We can however define a finite-memory strategy $\sigma'_0$ such that $\sigma'_0(v_0 v_2 v_3) = v_5$ and $\sigma'_0(v_0 v_2 v_3 v_5 v_3) = v_7$ and show that it is a solution to the problem. Indeed, the set of \paretoOptimalStrategy{$\sigma'_0$} payoffs is $\paretoSet{\sigma'_0} = \{(0, 1, 1),(1, 1, 0)\}$ and Player~$0$ wins every play consistent with $\sigma'_0$ whose payoff is in this set. A set $\Wit{\sigma'_0}$ of witnesses for these payoffs is $\{v_0v_2v_3v_5v_6^\omega, v_0v_2v_3v_5v_3v_7^\omega\}$ and is in this case the unique set of witnesses. This example shows that Player~$0$ sometimes needs memory in order to have a solution to the \problemAb{}.
\end{example}

\section{Fixed-Parameter Complexity}
\label{sec:newfpt}

In this section, we show that the \problemAb{} is in \FPT{} for both cases of reachability and parity \gamesAb{}. The details of our proof for each type of objective are provided separately in their own subsection. We refer the reader to~\cite{downey2012parameterized} for the concept of fixed-parameter complexity.

\begin{restatable}{theorem}{thmFPTreach} \label{thm:FPT-reach}
Solving the \problemAb{} is in \FPT{} for reachability \gamesAb{} for parameter~$\nbrObjectives$ equal to the number of objectives of Player~$1$.
\end{restatable}

\begin{restatable}{theorem}{thmFPTparity} \label{thm:FPT-parity}
Solving the \problemAb{} is in \FPT{} for parity \gamesAb{} for parameters $t$ and the maximal priority according to each parity objective of Player~$1$.
\end{restatable}

\subsection{Challenger-Prover Game} 
\label{subsec:cp}

In order to prove \autoref{thm:FPT-reach} and \autoref{thm:FPT-parity}, we provide a reduction to a specific two-player zero-sum game, called the \emph{\challengerProver{}} game (\challengerProverAb{} game). This game is a \emph{zero-sum}\footnote{We suppose the reader familiar with the concept of zero-sum games, see e.g.~\cite{2001automata}.} game played between \emph{Challenger} (written $\chal{}$) and \emph{Prover} (written $\prov{}$). We will show that Player~$0$ has a solution to the \problemAb{} in an \gameAb{} if and only if $\prov{}$ has a winning strategy in the corresponding \challengerProverAb{} game. In the latter game, $\prov{}$ tries to show the existence of a strategy $\sigma_0$ that is solution to the \problemAb{} in the original game and $\chal{}$ tries to disprove it. The \challengerProverAb{} game is described independently of the objectives used in the \gameAb{} and its objective is described as such in a \emph{generic way}. We later provide the proof of our \FPT{} results by adapting it specifically for reachability and parity \gameAb{}s.

\subparagraph*{Intuition on the \challengerProverAb{} Game.} 
Without loss of generality, the \gamesAb{} we consider in this section are such that each vertex in their arena has at most \emph{two successors}. It can be shown (see Appendix \ref{app:usefull_notions}) that any \gameAb{} $\mathcal{G}$ with $n$ vertices can be transformed into an \gameAb{} $\bar{\mathcal{G}}$ with $\mathcal{O}(n^2)$ vertices such that every vertex has at most two successors and Player~$0$ has a solution to the \problemAb{} in $\mathcal{G}$ if and only if he has a solution to the \problemAb{} in $\bar{\mathcal{G}}$.

Let $\mathcal{G}$ be an \gameAb{}. The \challengerProverAb{} game $\mathcal{G'}$ is a zero-sum game associated with $\mathcal{G}$ that intuitively works as follows. First, $\prov{}$ selects a set $P$ of payoffs which he announces as the set of Pareto-optimal payoffs $\paretoSet{\sigma_0}$ for the solution $\sigma_0$ to the \problemAb{} in $\mathcal{G}$ he is trying to construct. Then, $\prov{}$ tries to show that there exists a set of witnesses $\Wit{\sigma_0}$ in $\mathcal{G}$ for the payoffs in $P$. After the selection of $P$ in $\mathcal{G}'$, there is a one-to-one correspondence between plays in the arenas $G$ and $G'$ such that the vertices in $G'$ are augmented with a set $\provWit$ which is a subset of $P$. Initially $\provWit$ is equal to $P$ and after some history in $G'$, $\provWit$ contains payoff $p$ if the corresponding history in $G$ is prefix of the witness with payoff $p$ in the set $\Wit{\sigma_0}$ that $\prov{}$ is building. In addition, the objective $\Omega_{\prov{}}$ of $\prov{}$ is such that he has a winning strategy $\sigma_\prov$ in $\mathcal{G'}$ if and only if the set $P$ that he selected coincides with the set $\paretoSet{\sigma_0}$ for the corresponding strategy $\sigma_0$ in $\mathcal{G}$ and the latter strategy is a solution to the \problemAb{} in $\mathcal{G}$. A part of the arena of the \challengerProverAb{} game for \autoref{ex:example} with a positional winning strategy for $\prov{}$ highlighted in bold is illustrated in Figure~\ref{fig:example_challenger_prover}. 

\begin{figure} 
	\centering
		\resizebox{\textwidth}{!}{%
		\begin{tikzpicture}[->,>=stealth, shorten >=1pt,auto]
           
		\node[draw, circle, minimum size=0.8cm] (bot) at (9,11.5){$\bot$};
		\node[draw,rectangle, rounded corners=11pt, minimum size=0.8cm] (v0) at (10, 10){$v_0, P, \{p_1, p_2\}$};
		\node[draw, rectangle, minimum size=0.8cm] (v0sq) at (13.5, 10){$v_0, P, (\emptyset, \{p_1, p_2\})$};
		\node[draw,rectangle, rounded corners=11pt, minimum size=0.8cm] (v1) at (13.5, 11.5){$v_1, P, \emptyset$};
		\node[draw,rectangle, rounded corners=11pt, minimum size=0.8cm] (x) at (13.5, 8.5){$v_2, P, \{p_1, p_2\}$};

		\node[draw, rectangle, minimum size=0.8cm] (xsq) at (17, 8.5){$v_2, P, (\{p_1, p_2\}, \emptyset)$};
		\node[draw,rectangle, rounded corners=11pt, minimum size=0.8cm] (v2) at (17, 10){$v_3, P, \{p_1, p_2\}$};
		\node[draw,rectangle, rounded corners=11pt, minimum size=0.8cm] (v3) at (17, 7){$v_4, P, \emptyset$};

		\node[draw,rectangle, rounded corners=11pt, minimum size=0.8cm] (v4) at (20.5, 10){$v_5, P, \{p_1, p_2\}$};
		\node[draw,rectangle, rounded corners=11pt, minimum size=0.8cm] (v5) at (17, 11.5){$v_7, P, \{p_1, p_2\}$};
		\node[draw, rectangle, minimum size=0.8cm] (v4sq) at (24, 10){$v_5, P, (\{p_1\}, \{p_2\})$};
		
		\node[draw, rectangle, minimum size=0.8cm] (v4sqtre) at (24, 8.5){$v_5, P, (\{p_2\}, \{p_1\})$};
		\node[draw,rectangle, rounded corners=11pt, minimum size=0.8cm] (v2tre) at (24, 7){$v_3, P, \{p_2\}$};
		\node[draw,rectangle, rounded corners=11pt, minimum size=0.8cm] (v6tre) at (27.25, 8.5){$v_6, P, \{p_1\}$};
		\node[draw,rectangle, rounded corners=11pt, minimum size=0.8cm] (v5tre) at (27.25, 7){$v_7, P, \{p_2\}$};

		\node[draw, rectangle, minimum size=0.8cm] (v4sqfr) at (20.5, 8.5){$v_5, P, (\{p_1, p_2\}, \emptyset)$};
		\node[draw, rectangle, rounded corners=11pt, minimum size=0.8cm] (v6fr) at (20.5, 7){$v_6, P, \emptyset$};
				
    	\node[draw,rectangle, rounded corners=11pt, minimum size=0.8cm] (v2bis) at (24, 11.5){$v_3, P, \{p_1\}$};
		\node[draw,rectangle, rounded corners=11pt, minimum size=0.8cm] (v6) at (27.25, 10){$v_6, P, \{p_2\}$};
		
	    \node[draw,rectangle, rounded corners=11pt, minimum size=0.8cm] (v5bis) at (27.25, 11.5){$v_7, P, \{p_1\}$};

	    \node[rectangle, rounded corners=11pt, minimum size=0.8cm] (dotsbot) at (8, 10){$\dots$};
	    \node[rectangle, rounded corners=11pt, minimum size=0.8cm] (dotsv0) at (9, 8.5){$\dots$};
	    \node[rectangle, rounded corners=11pt, minimum size=0.8cm] (dotsv0bis) at (10, 8.5){$\dots$};
	    \node[rectangle, rounded corners=11pt, minimum size=0.8cm] (dotsv0tre) at (11, 8.5){$\dots$};
	    \node[rectangle, rounded corners=11pt, minimum size=0.8cm] (dotsx) at (13.5, 7){$\dots$};
	    \node[rectangle, rounded corners=11pt, minimum size=0.8cm] (dotsxbis) at (12.5, 7){$\dots$};
	    \node[rectangle, rounded corners=11pt, minimum size=0.8cm] (dotsxtre) at (14.5, 7){$\dots$};
	    \node[rectangle, rounded corners=11pt, minimum size=0.8cm] (dotsv4) at (20.5,11.5){$\dots$};
	    
	    \node[rectangle, rounded corners=11pt, minimum size=0.8cm] (dotsv3) at (24,12.75){$\dots$};
	    \node[rectangle, rounded corners=11pt, minimum size=0.8cm] (dotsv3bis) at (24,5.75){$\dots$};

	    \draw[-stealth] (bot) edge [] node {} (dotsbot);
		\draw[-stealth] (v0) edge [] node {} (dotsv0);
    	\draw[-stealth] (v0) edge [] node {} (dotsv0bis);
		\draw[-stealth] (v0) edge [] node {} (dotsv0tre);
		\draw[-stealth] (x) edge [] node {} (dotsx);
		\draw[-stealth] (x) edge [] node {} (dotsxbis);
		\draw[-stealth] (x) edge [] node {} (dotsxtre);
		\draw[-stealth] (v4) edge [] node {} (dotsv4);		
		\draw[-stealth] (v2bis) edge [] node {} (dotsv3);		
		\draw[-stealth] (v2tre) edge [] node {} (dotsv3bis);

		\draw[-stealth, very thick] (bot) edge [] node {} (v0);
		\draw[-stealth, very thick] (v0) edge [] node {} (v0sq);
		\draw[-stealth] (v0sq) edge [] node {} (v1);
		\draw[-stealth] (v0sq) edge [] node {} (x);
		
		\draw[-stealth, very thick] (x) edge [] node {} (xsq);
		\draw[-stealth] (xsq) edge [] node {} (v2);
		\draw[-stealth] (xsq) edge [] node {} (v3);
		
		\draw[-stealth, very thick] (v2) edge [] node {} (v4);
		\draw[-stealth] (v2) edge [] node {} (v5);
		\draw[-stealth, very thick] (v4) edge [] node {} (v4sq);
		
		\draw[-stealth] (v4sq) edge [] node {} (v2bis);
		\draw[-stealth] (v4sq) edge [] node {} (v6);
		
		\draw[-stealth, very thick] (v2bis) edge [] node {} (v5bis);
		
		\draw[-stealth, very thick] (v1) edge [loop above] node {} (v1);
		\draw[-stealth, very thick] (v5) edge [loop above] node {} (v5);
		\draw[-stealth, very thick] (v5bis) edge [loop above] node {} (v5bis);
		
		\draw[-stealth, very thick] (v3) edge [loop below] node {} (v3);
		\draw[-stealth, very thick] (v6) edge [loop above] node {} (v6);
		
		\draw[-stealth] (v4) edge [] node {} (v4sqtre);
		\draw[-stealth] (v4sqtre) edge [] node {} (v6tre);
		\draw[-stealth] (v4sqtre) edge [] node {} (v2tre);
		\draw[-stealth] (v2tre) edge [] node {} (v5tre);
		
		\draw[-stealth] (v4) edge [] node {} (v4sqfr);
		\draw[-stealth] (v4sqfr) edge [] node {} (v6fr);
		\draw[-stealth] (v4sqfr) edge [] node {} (v2);
		
		\draw[-stealth, very thick] (v6fr) edge [loop below] node {} (v6fr);

		\draw[-stealth] (v6tre) edge [loop below] node {} (v6tre);
		\draw[-stealth] (v5tre) edge [loop below] node {} (v5tre);

		\end{tikzpicture}
	}%
	\caption{A part of the \challengerProverAb{} game for Example \ref{ex:example} with $P = \{p_1, p_2\}, p_1 = (1,1,0)$ and $p_2 = (0,1,1)$. }
	\label{fig:example_challenger_prover}
\end{figure}
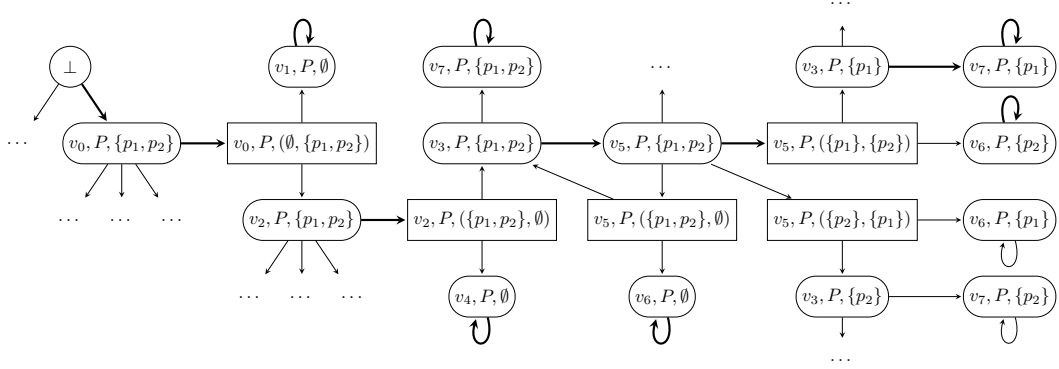 

\subparagraph*{Arena of the \challengerProverAb{} Game.} The initial vertex $\bot$ belongs to $\prov{}$. From this vertex, he selects a successor $(v_0, P, \provWit)$ such that $W = P$ and $P$ is an antichain of payoffs which $\prov$ announces as the set $\paretoSet{\sigma_0}$ for the strategy $\sigma_0$ in $G$ he is trying to construct. All vertices in plays starting with this vertex will have this same value for their $P$-component. Those vertices are either a triplet $(v, P, \provWit)$ that belongs to $\prov$ or $(v, P, (\provWit_l, \provWit_r))$ that belongs to~$\chal{}$. Given a play $\rho$ (resp.\ history $h$) in $G'$, we denote by $\rho_V$ (resp.\ $h_V$) the play (resp.\ history) in $G$ obtained by removing $\bot$ and keeping the $v$-component of every vertex of $\prov{}$ in $\rho$ (resp.\ $h$), which we call its \emph{projection}. 
\begin{itemize}
    \item After history $hm$ such that $m = (v, P, \provWit)$ with $v \in V_0$, $\prov{}$ selects a successor $v'$ such that $(v, v') \in E$ and vertex $(v', P, \provWit)$ is added to the play. This corresponds to Player~$0$ choosing a successor $v'$ after history $h_V v$ in~$G$. 
    \item After history $hm$ such that $m = (v, P, \provWit)$ with $v \in V_1$, $\prov{}$ selects a successor $(v, P, (\provWit_l, \provWit_r))$ with $(\provWit_l, \provWit_r)$ a partition of $\provWit$. This corresponds to $\prov{}$ splitting the set $W$ into two parts according to the two successors $v_l$ and $v_r$ of $v$. For the strategy $\sigma_0$ that $\prov{}$ tries to construct and its set of witnesses $\Wit{\sigma_0}$ he is building, he asserts that $\provWit_l$ (resp.\ $\provWit_r$) is the set of payoffs of the witnesses in $\prefStrat{h_V v_l}$ (resp.\ $\prefStrat{h_V v_r}$).
    \item From a vertex $(v, P, (\provWit_l, \provWit_r))$, $\chal{}$ can select a successor $(v_l, P, \provWit_l)$ or $(v_r, P, \provWit_r)$ which corresponds to the choice of Player~$1$. 
\end{itemize}   

\noindent
Formally, the game arena of the \challengerProverAb{} game is the tuple $G' = (V', V'_{\prov{}}, V'_{\chal{}}, E', \bot)$ with
\begin{itemize}
    \item $V'_{\prov{}} = \{ \bot \} \cup \{ (v, P, \provWit) \mid v \in V, P \subseteq \{0,1\}^\nbrObjectives \text{ is an antichain and } \provWit \subseteq P \}$,
    \item $V'_{\chal{}} = \{ (v, P, (\provWit_l, \provWit_r)) \mid v \in V_1, P \subseteq \{0,1\}^\nbrObjectives \text{ is an antichain and } \provWit_l, \provWit_r \subseteq P \}$,
    \item $(\bot, (v, P, \provWit)) \in E'$ if $v = v_0$ and $P = \provWit$,
    \item $((v, P, \provWit), (v', P, \provWit)) \in E'$ if $v \in V_0$ and $(v, v') \in E$,
    \item $((v, P, \provWit), (v, P, (\provWit_l, \provWit_r))) \in E'$ if $v \in V_1$ and $(\provWit_l, \provWit_r)$ is a partition of $\provWit$, 
    \item $((v, P, (\provWit_l, \provWit_r)), (v', P, \provWit) ) \in E'$ if $(v, v') \in E$ and $\{v' = v_l$ and $\provWit = \provWit_l\}$ or 
    $\{v' = v_r$ and $\provWit = \provWit_r\}$.
\end{itemize}
In the definition of $E'$, if $v$ has a single successor $v'$ in $G$, it is assumed to be $v_l$ and $W_r$ is always equal to $\emptyset$. We use as a convention that given the two successors $v_i$ and $v_j$ of vertex $v$, $v_i$ is the left successor if $i < j$.

\subparagraph*{Objective of $\prov$ in the \challengerProverAb{} Game.}
Let us now discuss the objective $\Omega_{\prov{}}$ of $\prov{}$. The $\provWit$-component of the vertices controlled by $\prov{}$ has a size that decreases along a play $\rho$ in $G'$. We write $lim_\provWit(\rho)$ the value of the $\provWit$-component at the limit in $\rho$. Recall that with this $\provWit$-component, $\prov{}$ tries to construct a solution $\sigma_0$ to the \problemAb{} with associated sets $\paretoSet{\sigma_0}$ and $\Wit{\sigma_0}$. Therefore, for him to win in the \challengerProverAb{} game, $lim_\provWit(\rho)$ must be a singleton or empty in every consistent play such that:
\begin{itemize}
    \item $lim_\provWit(\rho)$ must be a singleton $\{p\}$ with $p$ the payoff of $\rho_V$ in $G$, showing that $\rho_V \in \Wit{\sigma_0}$ is a correct witness for $p$. In addition, it must hold that $\won{\rho_V} = 1$ as $p \in P$ and as $\prov{}$ wants $\sigma_0$ to be a solution.
    \item $lim_\provWit(\rho)$ must be the empty set such that either the payoff of $\rho_V$ belongs to $\paretoSet{\sigma_0}$ and $\won{\rho_V} = 1$, or the payoff of $\rho_V$ is strictly smaller than some payoff in~$\paretoSet{\sigma_0}$.
\end{itemize}
These conditions verify that the sets $P = \paretoSet{\sigma_0}$ and $\Wit{\sigma_0}$ are correct and that $\sigma_0$ is indeed a solution to the \problemAb{} in $G$. They are generic as they do not depend on the actual objectives used in the \gameAb{}. 

Let us give the formal definition of $\Omega_{\prov{}}$. For an antichain $P$ of payoffs, we write $\Plays^P_{G'}$ the set of plays in $G'$ which start with $\bot (v_0, P, P)$ and we define the following set 
\begin{alignat}{3}
B_P = \big{\{} \rho \in \Plays^P_{G'} ~\mid~ &(lim_\provWit(\rho) = \{p\} &&\land \payoff{\rho_V} = p \in P &&\land \won{\rho_V} = 1) \ \lor \label{cp_cond_1}\\
                &(lim_\provWit(\rho) = \emptyset &&\land \payoff{\rho_V} \in P &&\land \won{\rho_V} = 1) \ \lor \label{cp_cond_2} \\
             &(lim_\provWit(\rho) = \emptyset &&\land \exists p \in P, \payoff{\rho_V} < && ~p) \big{\}}. \label{cp_cond_3}
\end{alignat}
Objective $\Omega_{\prov{}}$ of $\prov{}$ in $\mathcal{G'}$ is the union of $B_P$ over all antichains $P$. As the \challengerProverAb{} game is zero-sum, objective $\Omega_{\chal{}}$ equals $\Plays_{G'} \setminus \Omega_{\prov{}}$. The following theorem holds.

\begin{restatable}{theorem}{thmCP} \label{thm:CP} 
Player~$0$ has a strategy $\sigma_0$ that is solution to the \problemAb{} in $\mathcal{G}$ if and only if $\prov{}$ has a winning strategy $\sigma_{\prov{}}$ from $\bot$ in the \challengerProverAb{} game $\mathcal{G'}$.
\end{restatable}

\begin{proof}[Proof of \autoref{thm:CP}]
    Let us first assume that Player~$0$ has a strategy $\sigma_0$ that is solution to the \problemAb{}{} in $\mathcal{G}$. Let $\paretoSet{\sigma_0}$ be its set of \paretoOptimal\ payoffs and let $\Wit{\sigma_0}$ be a set of witnesses. We construct the strategy $\sigma_{\prov{}}$ from $\sigma_0$ such that
    \begin{itemize}
        \item $\sigma_{\prov{}}(\bot) = (v_0, P, P)$ such that $P = \paretoSet{\sigma_0}$ (this vertex exists as $\paretoSet{\sigma_0}$ is an antichain),
        \item $\sigma_{\prov{}}(hm) = (v', P, \provWit)$ if $m = (v, P, \provWit)$ with $v \in V_0$ and $v' = \sigma_0(h_V v)$, 
        \item $\sigma_{\prov{}}(hm) = (v, P, (\provWit_l, \provWit_r))$ if $m = (v, P, \provWit)$ with $v \in V_1$ and for $i \in\{l, r\}$, $\provWit_i = \{ \payoff{\rho} \mid \rho \in \prefStrat{h_V v_i}\}$.
    \end{itemize}   
    It is clear that given a play $\rho$ in $G'$ consistent with $\sigma_{\prov}$, the play $\rho_V$ in $G$ is consistent with $\sigma_0$. Let us show that $\sigma_{\prov}$ is winning for $\prov{}$ from $\bot$ in $G'$. Consider a play $\rho$ in $G'$ consistent with $\sigma_\prov$. There are two possibilities. \emph{(i)} $\rho_V$ is a witness of $\Wit{\sigma_0}$ and by construction $lim_\provWit(\rho) = \{p\}$ with $p = \payoff{\rho_V}$; thus $\won{\rho_V} = 1$ as $\sigma_0$ is a solution and $\rho_V$ is a witness. \emph{(ii)} $\rho_V$ is not a witness and by construction $lim_\provWit(\rho) = \emptyset$; as $\sigma_0$ is a solution, then $p = \payoff{\rho_V}$ is bounded by some payoff of $\paretoSet{\sigma_0}$ and in case of equality $\won{\rho_V} = 1$. Therefore $\rho$ satisfies the objective $B_P$ of $\Omega_{\prov{}}$ since it satisfies condition (\ref{cp_cond_1}) in case \emph{(i)} and condition (\ref{cp_cond_2}) or (\ref{cp_cond_3}) in case \emph{(ii)}. 

    Let us now assume that $\prov$ has a winning strategy $\sigma_{\prov}$ from $\bot$ in $G'$. Let $P$ be the antichain of payoffs chosen from $\bot$ by this strategy. We construct the strategy $\sigma_0$ from $\sigma_{\prov}$ such that $\sigma_0(h_Vv) = v'$ given $\sigma_{\prov}(hm) = (v', P, \provWit)$ with $m = (v, P, \provWit)$ and $v \in V_0$. Notice that this definition makes sense since there is a unique history $hm$ ending with a vertex of $\prov{}$ associated with $h_Vv$ showing a one-to-one correspondence between those histories. 
    
    Let us show $\sigma_0$ is a solution to the \problemAb{} with $\paretoSet{\sigma_0}$ being the set $P$. 
    First notice that $P$ is not empty. Indeed let $\rho$ be a play consistent with $\sigma_{\prov}$. As $\rho$ belongs to $\Omega_{\prov{}}$ and in particular to $B_P$, one can check that $P \neq \emptyset$ by inspecting conditions~(\ref{cp_cond_1}) to (\ref{cp_cond_3}). Second notice that by definition of $E'$, if $((v, P, \provWit), (v, P, (\provWit_l, \provWit_r))) \in E'$ with $\provWit \neq \emptyset$, then either $\provWit_l$ or $\provWit_r$ is not empty. Therefore given any payoff $p \in P$, there is a unique play $\rho$ consistent with $\sigma_{\prov}$ such that $lim_\provWit(\rho) = \{p\}$. By construction of $\sigma_0$ and as $\sigma_{\prov}$ is winning, the play $\rho_V$ is consistent with $\sigma_0$, has payoff $p$, and is won by Player~$0$ (see (\ref{cp_cond_1})). 
    
    Let $\rho_V$ be a play consistent with $\sigma_0$ and $\rho$ be the corresponding play consistent with $\sigma_{\prov}$. It remains to consider (\ref{cp_cond_2}) and (\ref{cp_cond_3}). These conditions indicate that $\rho_V$ has a payoff equal to or strictly smaller than a payoff in $P$ and that in case of equality $\won{\rho_V} = 1$. This shows that $\paretoSet{\sigma_0} = P$ and that $\sigma_0$ is a solution to the \problemAb{}.
\end{proof}

\subsection{Fixed-Parameter Complexity of Reachability SP Games} 
\label{subsec:fptreach}

We now develop the proof of \autoref{thm:FPT-reach} which works by specializing the generic objective $\Omega_{\prov{}}$ to handle reachability \gamesAb{}. We extend the arena $G'$ of the \challengerProverAb{} game such that its vertices keep track of the objectives of $\mathcal{G}$ which are satisfied along a play. Given an extended payoff $(w,p) \in \{0,1\} \times \{0,1\}^{\nbrObjectives}$ and a vertex $v \in V$, we define the \emph{payoff update} $\update{w,p,v} = (w',p')$ such that
$$\begin{array}{lll}
w' = 1   &\iff & w = 1 \text{ or } v \in \target_0,\\
p'_i = 1 &\iff& p_i = 1 \text{ or } v \in \target_i, \quad \forall i \in \{1, \ldots, \nbrObjectives\}.
\end{array}$$
We obtain the extended arena $G^*$ as follows: \emph{(i)} its set of vertices is $V' \times \{0,1\} \times \{0,1\}^\nbrObjectives$, \emph{(ii)} its initial vertex is $ \bot^* = (\bot, 0, (0,\ldots,0))$, and \emph{(iii)} $((m, w, p),(m', w', p'))$ with $m' = (v', P, \provWit)$ or $m' = (v', P, (\provWit_l, \provWit_r))$ is an edge in $G^*$ if $(m, m') \in E'$ and $(w',p') = \update{w,p,v'}$. 

We define the zero-sum game $\mathcal{G^*} = (G^*, \Omega^*_\prov)$ in which the three abstract conditions (\ref{cp_cond_1}-\ref{cp_cond_3}) detailed previously are encoded into the following B\"uchi objective by using the $(w, p)$-component added to vertices. We define $\Omega^*_\prov = \Buchi{B^*}$ with
\begin{alignat}{2}
B^* = \big{\{} (v, P, \provWit, w, p) \in V^*_{\prov{}} ~\mid~
            &(W = \{p\} &&\land  w = 1) \ \lor \tag{\ref{cp_cond_1}'} \\
            &(W = \emptyset  &&\land  p \in P  \land w = 1) \ \lor \tag{\ref{cp_cond_2}'} \\
            &(W = \emptyset  &&\land  \exists p' \in  P,~ p < p') \big{\}}. \tag{\ref{cp_cond_3}'}
\end{alignat}

\begin{restatable}{proposition}{propFPTreach}  \label{prop:FPTreach}
Player~$0$ has a strategy $\sigma_0$ that is solution to the \problemAb{} in a reachability \gameAb{} $\mathcal{G}$ if and only if $\prov{}$ has a winning strategy $\sigma^*_{\prov}$ in $\mathcal{G^*}$.
\end{restatable}

The proof of this proposition is a consequence of \autoref{thm:CP}. Using the one-to-one correspondence between plays in $G$ and plays in $G^*$ and the fact that $\mathcal{G}$ is a reachability \gameAb{}, the $(w,p)$-component in vertices of $G^*$ allows us to easily retrieve the extended payoff of a play in $G$. Indeed, in a play $\rho \in \Plays_{G^*}$, given the construction of $G^*$ and the payoff update function, it holds that from some point on the $W$- and $(w,p)$-components are constant. Therefore it holds that $w = \won{\rho_V}$, $p = \payoff{\rho_V}$ and $\provWit = lim_\provWit(\rho)$ for that play $\rho$. Moreover the {$P$-component} is constant along a play in $G^*$. 
It is direct to see that the plays $\rho$ in $G^*$ which visit infinitely often the set $B^*$, and therefore satisfy the B\"uchi objective $\Omega^*_\prov = \Buchi{B^*}$, satisfy one of the three conditions (\ref{cp_cond_1}-\ref{cp_cond_3}) stated in \autoref{subsec:cp}. The converse is also true. 

We now describe a \FPT{} algorithm for deciding the existence of a solution to the \problemAb{} in a reachability \gameAb{}, thus proving \autoref{thm:FPT-reach}.

\begin{proof}[Proof of \autoref{thm:FPT-reach}.]
We describe the following \FPT{} algorithm (for parameter~$\nbrObjectives$) for deciding the existence of a solution to the \problemAb{} in a reachability \gameAb{} $\mathcal{G}$ by using Proposition~\ref{prop:FPTreach}. First, we construct the zero-sum game $\mathcal{G^*}$. Its number $n$ of vertices is upper-bounded by $1 + |V| \cdot 2^{2^{\nbrObjectives+1}} \cdot 2^{\nbrObjectives+1} + |V| \cdot 2^{3 \cdot 2^t} \cdot 2^{\nbrObjectives+1}$. Indeed, except the initial vertex, vertices are of the form either $(v, P, \provWit,w,p)$ or $(v, P, (\provWit_l, \provWit_r),w,p)$ such that $P$, $\provWit$, $\provWit_l$ and $\provWit_r$ are antichains of payoffs in $\{0,1\}^\nbrObjectives$, and $(w,p)$ is an extended payoff. The construction of $\mathcal{G^*}$ is thus in \FPT{} for parameter~$\nbrObjectives$. Second, By Proposition~\ref{prop:FPTreach}, deciding whether there exists a solution to the \problemAb{} in $\mathcal{G}$ amounts to deciding if $\prov$ has a winning strategy from $\bot^*$ in $\mathcal{G^*}$. Since the objective $\Omega^*_\prov$ of $\prov$ in $\mathcal{G^*}$ is a B\"uchi objective, this game can be solved in $\mathcal{O}(n^2)$~\cite{ChatterjeeH14}. It follows that $\mathcal{G^*}$ can be solved in \FPT{} for parameter~$\nbrObjectives$.
\end{proof}

\subsection{Fixed-Parameter Complexity of Parity SP Games} 
\label{subsec:fptpar}

We now turn to parity \gamesAb{} and explain why solving the \problemAb{} in these games is in \FPT{}, again by reduction to the \challengerProverAb{} game. To this end, we first  recall the notion of Boolean B\"uchi games.

Boolean B\"uchi games are zero-sum games which we use in our reduction to the \challengerProverAb{} for parity \gameAb{}. Given $m$ sets $\target_1, \dots, \target_m$ such that $\target_i \subseteq V$, $i \in \{1, \dots, m\}$ and $\phi$ a Boolean formula over the set of variables $X = \{x_1, \dots, x_m\}$, the \emph{Boolean B\"uchi} objective $\BooleanBuchi{\phi, \target_1, \dots, \target_ m} = \{\rho \in \Plays \mid \rho \text{ satisfies } (\phi, \target_1, \dots, \target_ m) \}$ is the set of plays whose valuation of the variables in $X$ satisfy formula $\phi$. Given a play $\rho$, its valuation is such that $x_i = 1$ if and only if $\infOcc{\rho} \cap \target_i \neq  \emptyset$ and $x_i = 0$ otherwise. That is, a play satisfies the objective if the Boolean formula describing sets to be visited infinitely often by a play is satisfied. We denote by $|\phi|$ the size of $\phi$ as equal to the number of conjunctions and disjunctions in $\phi$. The following theorem on the fixed-parameter complexity of Boolean B\"uchi games is proved in~\cite{BruyereHR18}.

\begin{theorem}
\label{bb_complexity}
    Solving Boolean B\"uchi games is in \FPT, with an algorithm in $\mathcal{O}(2^M \cdot |\phi|+ (M^M \cdot |V|)^5)$ time with $M = 2^m$ such that $m$ is the number of variables and $|\phi|$ is the size of $\phi$ in the Boolean B\"uchi objective \cite{BruyereHR18}.
\end{theorem}

Let $\mathcal{G} = (G, \Omega_0, \dots, \Omega_\nbrObjectives)$ be a parity \gameAb{} with parity objectives such that $\Omega_i = \parity{c_i}$ for a priority function $c_i : V \rightarrow \N$. Let $G'$ be the arena of the \challengerProverAb{} game. In the following, we construct a Boolean B\"uchi objective $\Omega'_\prov$ for $\prov$ such that the following proposition holds.

\begin{proposition} \label{prop:FPTparity}
Player~$0$ has a strategy $\sigma_0$ that is solution to the \problemAb{} in $\mathcal{G}$ if and only if $\prov{}$ has a winning strategy $\sigma_{\prov}$ in $\mathcal{G^*} = (G', \Omega'_\prov)$.
\end{proposition}

\begin{proof}[Proof of \autoref{prop:FPTparity}.]
let $G' = (V', V'_{\prov{}}, V'_{\chal{}}, E', \bot)$ be the arena of the \challengerProverAb{} game presented in Section~\ref{subsec:cp}. Recall that the objective $\Omega_\prov$ of this game is the union of the sets $B_P$ over all antichains $P$ such that $B_P$ is the disjunction of conditions (\ref{cp_cond_1}-\ref{cp_cond_3}). The idea of the proof is to translate this objective into a Boolean B\"uchi objective $\Omega'_\prov$. We will proceed step by step. The required Boolean formula for defining $\Omega'_\prov$ is equal to
$$ \phi  = \bigvee \limits_{P} \big{(} x_{P} \land (\mathit{cond^{P}_1} \lor \mathit{cond^{P}_2} \lor \mathit{cond^{P}_3}) \big{)}$$
such that the main disjunction is over all antichains $P$. The variable $x_{P}$ corresponds to the set $T_{P} = \{ (v, P, \provWit) \in V'_{\prov}\}$. The valuation of $x_P$ is true for a given play if and only if the set $T_{P}$ is visited infinitely often and therefore $P$ is the antichain chosen by $\prov$ in $G'$. Since the $P$-component is constant along a play, only one $x_P$ is valued as true for a given play. Let us now detail each subformula $\mathit{cond^{P}_i}$ that is the translation of condition $(i)$ for $i \in \{1,2,3\}$.

Let us begin with the encoding of payoffs. Let $d_0, \dots, d_\nbrObjectives$ be such that $d_i$ is the maximal even priority appearing in $G$ according to priority function $c_i$ for objective $\Omega_i$ with $i \in \{0, \dots, \nbrObjectives\}$. 

First, we show that a parity objective $\parity{c_i}$ from $\mathcal{G}$ can be encoded as a Boolean B\"uchi objective. Given the parity objective $\parity{c_i}$, we construct the Boolean formula $\mathit{parity_i}$ over variables $\{x^i_0, x^i_1, \dots, x^i_{d_{i}}\}$ such that $$\mathit{parity_i} = x^i_0 \lor (x^i_2 \land \neg x^i_1) \lor \dots \lor (x^i_{d_i} \land \neg x^i_{d_i-1} \land \neg x^i_{d_i - 3} \land \dots \land \neg x^i_1)$$ and for $j \in \{0, \dots, d_i\}$, the set corresponding to variable $x^i_j$ is $T^i_j = \{ (v, P, \provWit) \in V'_{\prov} \mid c_i(v) = j \}$. It is easy to show that the parity objective $\parity{c_i}$ is satisfied if and only if the Boolean B\"uchi objective $\BooleanBuchi{parity_i, T^i_0, \dots, T^i_{d_i}}$ is satisfied.
    
Second, given a payoff $p = (p_1, \dots, p_\nbrObjectives)$ in $\mathcal{G}$, we consider the Boolean formula
$$\mathit{payoff_p} = C_1 \land \dots \land C_\nbrObjectives$$ 
such that $C_i = \mathit{parity_i}$ if $p_i = 1$ and $C_i = \neg \mathit{parity_i}$ otherwise. 
Clearly the projection $\rho_V$ of play $\rho$ realizes payoff $p$ if and only if $\rho$ satisfies the Boolean B\"uchi objective $\BooleanBuchi{\mathit{payoff_p}, T^1_0, \dots, T^1_{d_1},\dots, T^\nbrObjectives_0, \dots, T^\nbrObjectives_{d_\nbrObjectives}}$.

We now fix some antichain $P$. Let us detail subformula $\mathit{cond^{P}_1}$ encoding condition (\ref{cp_cond_1}). For a payoff $p$, we define formula $$single_p = x_p \land \mathit{payoff_p} \land \mathit{parity_0}$$ such that $x_p = \{ (v, P, \provWit) \in V'_{\prov} \mid W = \{p\} \}$. 
Since at some point during a play, the $\provWit$-component stabilizes and since a play $\rho$ which satisfies $\mathit{payoff_p} \land \mathit{parity_0}$ is such that $\payoff{\rho_V} = p$ and $\won{\rho_V} = 1$, it holds that satisfying this formula corresponds exactly to satisfying condition (\ref{cp_cond_1}) for some $p$. Formula $\mathit{cond^{P}_1}$ is thus the disjunction 
$$\mathit{cond^{P}_1} = \bigvee \limits_{p \in P} \mathit{single_p}.$$ 
 
Let us shift to subformula $\mathit{cond^{P}_2}$ encoding condition (\ref{cp_cond_2}).
Using similar arguments, we define for payoff $p$ formula $$empty_p = x_\emptyset \land \mathit{payoff_p} \land \mathit{parity_0}$$ such that $x_\emptyset = \{ (v, P, \emptyset) \in V'_{\prov{}}\}$. It corresponds exactly to the set of plays $\rho$ such that $lim_\provWit(\rho) = \emptyset$, $\payoff{\rho_V} = p$ and $\won{\rho_V} = 1$. Therefore
$$\mathit{cond^{P}_2} = \bigvee \limits_{p \in P} \mathit{empty_p}.$$

Finally, we define subformula $\mathit{cond^{P}_3}$ encoding condition (\ref{cp_cond_3}). Let $\overline{P}$ be the set containing every payoff $p'$ such that $\exists p \in P, ~p' < p$. We define
$$\mathit{cond^{P}_3} = \bigvee \limits_{p' \in \overline{P}} \mathit{smaller_{p'}} $$
with $smaller_{p'}$ being the formula $x_\emptyset \land \mathit{payoff_{p'}}$.

Notice that the Boolean formula $\phi$ constructed in this proof has a number $m$ of variables and a size $|\phi|$ that only depend on $\nbrObjectives$ and $d_i$, $i \in \{0,\ldots,\nbrObjectives\}$.
\end{proof}

This previous construction can be used to provide the proof of \autoref{thm:FPT-parity}.

\begin{proof}[Proof of \autoref{thm:FPT-parity}.]
We describe the following \FPT{} algorithm for deciding the existence of a solution to the \problemAb{} in a parity \gameAb{} $\mathcal{G}$ by using Proposition~\ref{prop:FPTparity}. First, we construct the zero-sum game $\mathcal{G^*}$ of Proposition~\ref{prop:FPTparity}. Its number $n$ of vertices is upper-bounded by $1 + |V| \cdot 2^{2^{\nbrObjectives+1}} + |V| \cdot 2^{3 \cdot 2^{\nbrObjectives}}$ which is in $\mathcal{O}(|V| \cdot f(\nbrObjectives))$ with $f$ a computable function which only depends on $\nbrObjectives$. Moreover the number $m$ of variables and the size $|\phi|$ of the Boolean formula $\phi$ defining the Boolean B\"uchi objective of $\mathcal{G^*}$ depend only on parameters $\nbrObjectives$ and $d_i$ for $i \in \{0, \dots, \nbrObjectives\}$. Therefore the construction of $\mathcal{G^*}$ is in \FPT{} for these parameters. Deciding whether there exists a solution to the \problemAb{} in $\mathcal{G}$ amounts to deciding if $\prov$ has a winning strategy from $\bot$ in $\mathcal{G^*}$. By \autoref{bb_complexity}, the latter Boolean B\"uchi game can be solved with an algorithm in $\mathcal{O}(2^M \cdot |\phi|+ (M^M \cdot n)^5)$ time with $M = 2^m$. It follows that $\mathcal{G^*}$ can be solved in $\mathcal{O}(2^M \cdot |\phi|+ (M^M \cdot |V| \cdot f(\nbrObjectives))^5)$ which is in \FPT{} for the announced parameters.
\end{proof}

\section{Complexity Class of the SPS Problem}
\label{sec:nexptimecomplete}

In this section, we study the complexity class of the \problemAb{} and prove its \nexptime{}-completeness for both reachability and parity \gamesAb{}. 

\subsection{\textsf{NEXPTIME}-Membership}
\label{subsec:nexptime_member}

We first show the  membership to \nexptime{} of the \problemAb{} by providing a nondeterministic algorithm with time exponential in the size of the game $\mathcal{G}$. By \emph{size}, we mean the number $|V|$ of its vertices and the number $\nbrObjectives$ of objectives of Player~$1$. Notice that the time complexity of the \FPT{} algorithms obtained in the previous section is too high, preventing us from directly using the \challengerProverAb{} game to show a tight membership result. Conversely, the nondeterministic algorithm provided in this section is not \FPT{} as it is exponential in $|V|$.

\begin{restatable}{theorem}{thmnexptime}  \label{thm:nexptime}
The \problemAb{} is in \nexptime{} for reachability and parity \gamesAb{}. 
\end{restatable}

We show that the \problemAb{} is in \nexptime{} by proving that if Player~$0$ has a strategy which is a solution to the problem, then he has one which is finite-memory with at most an exponential number of memory states\footnote{Recall that to have a solution to the \problemAb{}, memory is sometimes necessary as shown in Example~\ref{ex:example}.}.
This yields a \nexptime{} algorithm in which we nondeterministically guess such a strategy and check in exponential time that it is indeed a solution to the problem. 

\begin{proposition} \label{prop:tildesigma}
Let $\mathcal{G}$ be a reachability \gameAb{} or a parity \gameAb{}. Let $\sigma_0$ be a solution the the \problemAb{}. Then there exists another solution $\tilde{\sigma}_0$ that is finite-memory and has a memory size exponential in the size of $\mathcal{G}$.
\end{proposition}

While the proof of Proposition~\ref{prop:tildesigma} requires some specific arguments to treat both reachability and parity objectives, it is based on the following common principles. 
\begin{itemize}
    \item We start from a winning strategy $\sigma_0$ for the \problemAb{} and the objectives $\ObjPlayer{0},\ObjPlayer{1},\dots,\ObjPlayer{\nbrObjectives}$ and consider a set of witnesses $\Wit{\sigma_0}$, that contains one play for each element of the set  $\paretoSet{\sigma_0}$ of \paretoOptimal{} payoffs.
    \item We start by showing the existence of a strategy $\hat{\sigma}_0$ constructed from $\sigma_0$, in which Player~$0$ follows $\sigma_0$ as long as the current consistent history is prefix of at least one witness in $\Wit{\sigma_0}$. Then when a deviation from $\Wit{\sigma_0}$ occurs, Player~$0$ switches to a so-called \emph{punishing strategy}. A deviation is a history that leaves the set of witnesses $\Wit{\sigma_0}$ after a move of Player~$1$ (this is not possible by a move of Player~$0$). After such a deviation, $\hat{\sigma}_0$ systematically imposes that the consistent play either satisfies $\ObjPlayer{0}$ or is not \paretoOptimal{}, i.e., it gives to Player~1 a payoff that is strictly smaller than the payoff of a witness in $\Wit{\sigma_0}$. This makes the deviation \emph{irrational} for Player~1.  We show that this can be done, both for reachability and parity objectives, with at most exponentially many different punishing strategies, each having a size bounded exponentially in the size of the game. The strategy $\hat{\sigma}_0$ that we obtain is therefore composed of the part of $\sigma_0$ that produces $\Wit{\sigma_0}$ and a punishment part whose size is at most exponential.
    \item Then, we show how to decompose each witness in $\Wit{\sigma_0}$ into at most exponentially many \emph{sections} that can, in turn, be compacted into finite elementary paths or lasso shaped paths of polynomial length. As $\Wit{\sigma_0}$ contains exactly $|\paretoSet{\sigma_0}|$ witnesses $\rho$, those compact witnesses $c\rho$ can be produced by a finite-memory strategy with an exponential size for both reachability and parity objectives. This allows us to construct a strategy $\tilde{\sigma}_0$ that produces the compact witnesses and acts as $\hat{\sigma}_0$ after any deviation. This strategy is a solution of the \problemAb{} and has an exponential size as announced.
\end{itemize}

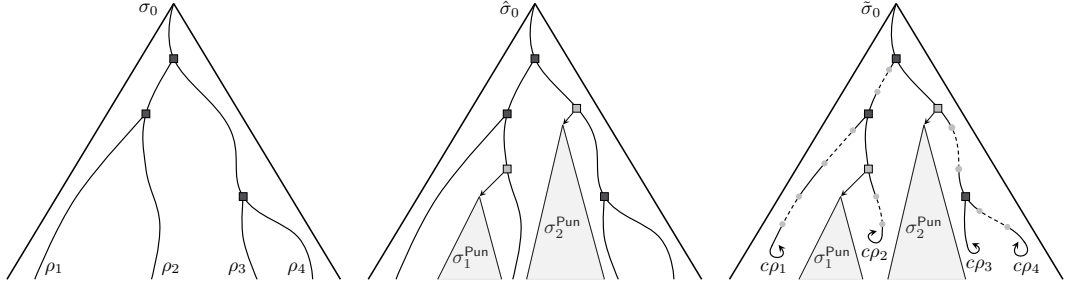
\begin{figure}
	\centering
		\resizebox{\textwidth}{!}{%
		\begin{tikzpicture}
          

        \draw[-, thick] (10,10) edge [] node {} (7, 5);
        \draw[-, thick] (10,10) edge [] node {} (13, 5);

        \node[] (0) at (9.55,9.9){$\sigma_0$};
        \node[] (0) at (10,10){};
        
        \node[rectangle, fill=darkgray, inner sep = 0pt, minimum size=4pt] (1) at (10,9){};
        \node[rectangle, fill=darkgray, inner sep = 0pt, minimum size=4pt] (2) at (9.5,8){};
        \node[rectangle, fill=darkgray, inner sep = 0pt, minimum size=4pt] (3) at (11.25,6.5){};
        
        \node[] (4) at (7.5,5){};
        \node[] () at (7.85,5.2){$\rho_1$};
        \node[] (5) at (9.6,5){};
        \node[] () at (9.95,5.2){$\rho_2$};
        \node[] (6) at (11.5,5){};
        \node[] () at (11.15,5.2){$\rho_3$};
        \node[] (7) at (12.5,5){};
        \node[] () at (12.22,5.2){$\rho_4$};

		\draw[semithick] plot [smooth, tension=0.7] coordinates { (0) (1) (11,7.75) (3) (12.25, 5.75) (7)};
		\draw[semithick] plot [smooth, tension=0.7] coordinates { (3) (11.25,5.75) (6)};
		\draw[semithick] plot [smooth, tension=0.7] coordinates { (1) (2) (9.5,7) (9.75,6) (5)};
		\draw[semithick] plot [smooth, tension=0.7] coordinates { (2) (8.25, 6.5) (4)};
		
        \node[draw, rectangle, fill=darkgray, inner sep = 0pt, minimum size=4pt] () at (10,9){};
        \node[draw, rectangle, fill=darkgray, inner sep = 0pt, minimum size=4pt] () at (9.5,8){};
        \node[draw, rectangle, fill=darkgray, inner sep = 0pt, minimum size=4pt] () at (11.25,6.5){};

        
        \draw[-, thick] (16.5,10) edge [] node {} (13.5, 5);
        \draw[-, thick] (16.5,10) edge [] node {} (19.5, 5);

        \node[] (0bis) at (16.05,9.9){$\hat{\sigma}_0$};
        \node[] (0bis) at (16.5,10){};
        
        \node[rectangle, fill=darkgray, inner sep = 0pt, minimum size=4pt] (1bis) at (16.5,9){};
        \node[rectangle, fill=darkgray, inner sep = 0pt, minimum size=4pt] (2bis) at (16,8){};
        \node[rectangle, fill=darkgray, inner sep = 0pt, minimum size=4pt] (3bis) at (17.75,6.5){};
        
        \node[] (4bis) at (14,5){};
        \node[] (5bis) at (16.1,5){};
        \node[] (6bis) at (18,5){};
        \node[] (7bis) at (19,5){};

		\draw[semithick] plot [smooth, tension=0.7] coordinates { (0bis) (1bis) (17.5,7.75) (3bis) (18.75, 5.75) (7bis)};
		\draw[semithick] plot [smooth, tension=0.7] coordinates { (3bis) (17.75,5.75) (6bis)};
		\draw[semithick] plot [smooth, tension=0.7] coordinates { (1bis) (2bis) (16,7) (16.25,6) (5bis)};
		\draw[semithick] plot [smooth, tension=0.7] coordinates { (2bis) (14.75, 6.5) (4bis)};

        \node[draw, rectangle, fill=darkgray, inner sep = 0pt, minimum size=4pt] () at (16.5,9){};
        \node[draw, rectangle, fill=darkgray, inner sep = 0pt, minimum size=4pt] () at (16,8){};
        \node[draw, rectangle, fill=darkgray, inner sep = 0pt, minimum size=4pt] () at (17.75,6.5){};
        
		\node[draw, rectangle, fill=lightgray, inner sep = 0pt, minimum size=4pt] (pun1bis) at (16,7){};
		\node[inner sep = 0pt, minimum size=0pt] (pun1bisroot) at (15.5,6.5){};
        \draw[-] (15.5,6.5) edge [] node {} (14.75, 5);
        \draw[-] (15.5,6.5) edge [] node {} (15.9, 5);
		\node[] () at (15.35,5.4){$\punStrat{1}$};
        \draw[-stealth] (pun1bis) edge [] node {} (pun1bisroot);
        \draw [draw=gray, fill=lightgray, opacity=0.2]
       (15.5,6.5) -- (14.75, 5) -- (15.9, 5) --  cycle;

        \node[draw, rectangle, fill=lightgray, inner sep = 0pt, minimum size=4pt] (pun2bis) at (17.25,8.1){}; 
		\node[inner sep = 0pt, minimum size=0pt] (pun2bisroot) at (17,7.8){};
        \draw[-] (17,7.8) edge [] node {} (16.35, 5);
        \draw[-] (17,7.8) edge [] node {} (17.75, 5);
		\node[] () at (17,5.95){$\punStrat{2}$};
        \draw[-stealth] (pun2bis) edge [] node {} (pun2bisroot);
        \draw [draw=gray, fill=lightgray, opacity=0.2]
       (17,7.8) -- (16.35, 5) -- (17.75, 5) --  cycle;
       
        \draw[-, thick] (23,10) edge [] node {} (20, 5);
        \draw[-, thick] (23,10) edge [] node {} (26, 5);

        \node[] (0tre) at (22.55,9.9){$\tilde{\sigma}_0$};
        \node[] (0tre) at (23,10){};
        
        \node[rectangle, fill=darkgray, inner sep = 0pt, minimum size=4pt] (1tre) at (23,9){};
        \node[rectangle, fill=darkgray, inner sep = 0pt, minimum size=4pt] (2tre) at (22.5,8){};
        \node[rectangle, fill=darkgray, inner sep = 0pt, minimum size=4pt] (3tre) at (24.25,6.5){};
        
        \node[] (4tre) at (20.5,5){};
        \node[] (5tre) at (22.6,5){};
        \node[] (6tre) at (24.5,5){};
        \node[] (7tre) at (25.5,5){};

		\draw[semithick] plot [smooth, tension=0.7] coordinates { (0tre) (1tre) (24,7.75) (3tre) (25.25, 5.75) (7tre)};
		\draw[semithick] plot [smooth, tension=0.7] coordinates { (3tre) (24.25,5.75) (6tre)};
		\draw[semithick] plot [smooth, tension=0.7] coordinates { (1tre) (2tre) (22.5,7) (22.75,6) (5tre)};
		\draw[semithick] plot [smooth, tension=0.7] coordinates { (2tre) (21.25, 6.5) (4tre)};

        \node[draw, rectangle, fill=darkgray, inner sep = 0pt, minimum size=4pt] () at (23,9){};
        \node[draw, rectangle, fill=darkgray, inner sep = 0pt, minimum size=4pt] () at (22.5,8){};
        \node[draw, rectangle, fill=darkgray, inner sep = 0pt, minimum size=4pt] () at (24.25,6.5){};
        
		\node[draw, rectangle, fill=lightgray, inner sep = 0pt, minimum size=4pt] (pun1tre) at (22.5,7){};
		\node[inner sep = 0pt, minimum size=0pt] (pun1treroot) at (22,6.5){};
        \draw[-] (22,6.5) edge [] node {} (21.25, 5);
        \draw[-] (22,6.5) edge [] node {} (22.4, 5);
		\node[] () at (21.85,5.4){$\punStrat{1}$};
        \draw[-stealth] (pun1tre) edge [] node {} (pun1treroot);
        \draw[draw=gray, fill=lightgray, opacity=0.2]
       (22,6.5) -- (21.25, 5) -- (22.4, 5) --  cycle;

        \node[draw, rectangle, fill=lightgray, inner sep = 0pt, minimum size=4pt] (pun2tre) at (23.75,8.1){}; 
		\node[inner sep = 0pt, minimum size=0pt] (pun2treroot) at (23.5,7.8){};
        \draw[-] (23.5,7.8) edge [] node {} (22.85, 5);
        \draw[-] (23.5,7.8) edge [] node {} (24.25, 5);
		\node[] () at (23.5,5.95){$\punStrat{2}$};
        \draw[-stealth] (pun2tre) edge [] node {} (pun2treroot);
        \draw[draw=gray, fill=lightgray, opacity=0.2]
       (23.5,7.8) -- (22.85, 5) -- (24.25, 5) --  cycle;
        
        \node[circle, color=lightgray, fill=lightgray, inner sep = 0pt, minimum size=3pt] (cycle1) at (24,7.75){};
        \node[ circle, color=lightgray, fill=lightgray, inner sep = 0pt, minimum size=3pt] (cycle2) at (24.125,7){};
        \draw[ultra thick, dotted, draw=white] plot [smooth, tension=0.7] coordinates { (cycle1) (24.09, 7.5)  (cycle2)};
                
        \node[draw, circle, color=lightgray, fill=lightgray, inner sep = 0pt, minimum size=3pt] () at (24,7.75){};
        \node[draw, circle, color=lightgray, fill=lightgray, inner sep = 0pt, minimum size=3pt] () at (24.125,7){};
        
        \node[circle, color=lightgray, fill=lightgray, inner sep = 0pt, minimum size=3pt] (cycle3) at (24.5,6.24){};
        \node[ circle, color=lightgray, fill=lightgray, inner sep = 0pt, minimum size=3pt] (cycle4) at (25,5.95){};
        \draw[ultra thick, dotted, draw=white] plot [smooth, tension=0.7] coordinates { (cycle3) (cycle4)};
        \node[circle, color=lightgray, fill=lightgray, inner sep = 0pt, minimum size=3pt] () at (24.5,6.24){};
        \node[ circle, color=lightgray, fill=lightgray, inner sep = 0pt, minimum size=3pt] () at (25,5.95){};    
        
        \node[circle, color=lightgray, fill=lightgray, inner sep = 0pt, minimum size=3pt] (cycle5) at (22.64,6.5){};
        \node[ circle, color=lightgray, fill=lightgray, inner sep = 0pt, minimum size=3pt] (cycle6) at (22.75,6){};
        \draw[ultra thick, dotted, draw=white] plot [smooth, tension=0.7] coordinates { (cycle5) (22.65,6.45) (22.7,6.35) (cycle6)};
        \node[circle, color=lightgray, fill=lightgray, inner sep = 0pt, minimum size=3pt] (cycle5) at (22.64,6.5){};
        \node[ circle, color=lightgray, fill=lightgray, inner sep = 0pt, minimum size=3pt] (cycle6) at (22.75,6){};
        
        \node[circle, color=lightgray, fill=lightgray, inner sep = 0pt, minimum size=3pt] (cycle7) at (21.25,6.5){};
        \node[ circle, color=lightgray, fill=lightgray, inner sep = 0pt, minimum size=3pt] (cycle8) at (20.955,6){};
        \draw[ultra thick, dotted, draw=white] plot [smooth, tension=0.7] coordinates { (cycle7) (cycle8)};
        \node[circle, color=lightgray, fill=lightgray, inner sep = 0pt, minimum size=3pt] (cycle7) at (21.25,6.5){};
        \node[ circle, color=lightgray, fill=lightgray, inner sep = 0pt, minimum size=3pt] (cycle8) at (20.955,6){};
        
        \node[circle, color=lightgray, fill=lightgray, inner sep = 0pt, minimum size=3pt] (cycle7) at (22.22,7.7){};
        \node[ circle, color=lightgray, fill=lightgray, inner sep = 0pt, minimum size=3pt] (cycle8) at (21.71,7.1){};
        \draw[ultra thick, dotted, draw=white] plot [smooth, tension=0.7] coordinates { (cycle7) (cycle8)};
        \node[circle, color=lightgray, fill=lightgray, inner sep = 0pt, minimum size=3pt] (cycle7) at (22.22,7.7){};
        \node[ circle, color=lightgray, fill=lightgray, inner sep = 0pt, minimum size=3pt] (cycle8) at (21.71,7.1){};
        
        \node[circle, color=lightgray, fill=lightgray, inner sep = 0pt, minimum size=3pt] (cycle9) at (22.872,8.8){};
        \node[ circle, color=lightgray, fill=lightgray, inner sep = 0pt, minimum size=3pt] (cycle10) at (22.662,8.4){};
        \draw[ultra thick, dotted, draw=white] plot [smooth, tension=0.7] coordinates { (cycle9) (cycle10)};
        \node[circle, color=lightgray, fill=lightgray, inner sep = 0pt, minimum size=3pt] (cycle9) at (22.872,8.8){};
        \node[ circle, color=lightgray, fill=lightgray, inner sep = 0pt, minimum size=3pt] (cycle10) at (22.662,8.4){};
        
        
        \node[rectangle, inner sep = 0pt, minimum size=4pt] (loop4) at (25.243, 5.76){};
        \draw[draw=white, fill=white]
       (25.29, 5.645) -- (25.4, 5.645) -- (25.58,4.99)  -- (25.43,4.99) -- cycle;
        \draw[-stealth, semithick, out=-53,in=-120, looseness=10] (loop4) edge [] node {} (loop4);
        
        \node[rotate=20, rectangle, inner sep = 0pt, minimum size=4pt] (loop3) at (24.265, 5.7){};
        \draw[draw=white, fill=white]
       (24.35, 5.6) -- (24.23, 5.6) --  (24.43,4.99) -- (24.6,4.99) -- cycle;
        \draw[-stealth, semithick, out=-85,in=-30, looseness=10] (loop3) edge [] node {} (loop3);
        
        \node[ rotate=-30, rectangle, inner sep = 0pt, minimum size=4pt] (loop2) at (22.757, 5.9){};
        \draw[draw=white, fill=white]
       (22.68, 5.8) -- (22.82, 5.8) -- (22.7,4.99) -- (22.55,4.99) -- cycle;
        \draw[-stealth, semithick, out=-90,in=-160, looseness=10] (loop2) edge [] node {} (loop2);
        
        \node[rectangle, inner sep = 0pt, minimum size=4pt] (loop1) at (20.823, 5.7){};
        \draw[draw=white, fill=white]
       (20.6, 5.6) -- (20.8, 5.6) -- (20.6,4.99) -- (20.4,4.99) -- cycle;
        \draw[-stealth, semithick, rotate=-27, out=-100,in=-30, looseness=10] (loop1) edge [] node {} (loop1);
        
        \node[] () at (20.8,5.2){$c \rho_1$};
        \node[] () at (22.615,5.45){$c \rho_2$};
        \node[] () at (24.5,5.25){$c \rho_3$};
        \node[] () at (25.35,5.2){$c \rho_4$};
    
		\end{tikzpicture}
	}%
	\caption{The creation of strategies $\hat{\sigma}_0$ and $\tilde{\sigma_0}$ from a solution $\sigma_0$ with $\Wit{\sigma_0} = \{\rho_1, \rho_2, \rho_3, \rho_4\}$.}
	\label{fig:nexptime_strategy}
\end{figure} 

We now develop the details of the construction of the strategies $\hat{\sigma}_0$ and $\tilde{\sigma}_0$. Figure~\ref{fig:nexptime_strategy} illustrates this construction. It is done in several steps to finally get the proof of Proposition~\ref{prop:tildesigma}. For the rest of this section, we fix an \gameAb{} $\mathcal{G}$ with objectives $\ObjPlayer{0},\ObjPlayer{1},\dots,\ObjPlayer{t}$, a strategy $\sigma_0$ that is solution to the \problemAb{}, a set of witnesses $\Wit{\sigma_0}$ for the \paretoOptimal{} payoffs in $P_{\sigma_0}$, and we write $\dominatedPlays{\sigma_0}$ the set of plays whose payoff is strictly smaller than some payoff in $\paretoSet{\sigma_0}$.

\subparagraph*{Deviations and Punishing Strategies.}
First, we define the set of deviations $\deviations{\sigma_0}$ as follows:
$$\deviations{\sigma_0} = \{hv \in \Histsigma{\sigma_0} \mid \prefStrat{h} \neq \emptyset \land \prefStrat{hv} = \emptyset \}.$$
As explained above, a deviation is a history that leaves the set of witnesses $\Wit{\sigma_0}$ (by a move of Player~$1$).

Second, we establish the existence of canonical forms for punishing strategies. We potentially need an exponential number of them for reachability objectives and a polynomial number of them for parity objectives. In both cases, each punishing strategy has a size which can be bounded exponentially. The existence of those strategies are direct consequences of the two following lemmas.

\begin{lemma}[Parity] 
\label{lem:parity-case}
Let $v \in V$ be such that there exists $hv \in \deviations{\sigma_0}$.

\noindent
Then there exists a finite-memory strategy $\punStrat{v}$ such that for all deviations $hv \in \deviations{\sigma_0}$, when Player~$0$ plays $\punStrat{v}$ from $hv$, all consistent plays $\rho$ starting in $v$ are such that either $h\rho \in \ObjPlayer{0}$ or $h\rho \in \dominatedPlays{\sigma_0}$. The size of $\punStrat{v}$ is at most exponential in the size of $\mathcal{G}$.
\end{lemma}

\begin{proof}[Proof of \autoref{lem:parity-case}.]
First, we note that, after a deviation $hv \in \deviations{\sigma_0}$, if Player~$0$ continues to play the strategy $\sigma_0$ from $hv$, then all consistent plays $\rho$ are such that either $\rho \in \ObjPlayer{0}$ or $\rho \in \dominatedPlays{\sigma_0}$ as $\sigma_0$ is a solution to the \problemAb{}. Therefore, we know that Player~$0$ has a punishing strategy for all such deviations $hv$. Second, as parity objectives are prefix-independent, he can use one uniform strategy that only depends on $v$ (and not on $hv$). There exists such a strategy with finite memory that can be constructed as follows. We express the objective $\ObjPlayer{0} \cup \dominatedPlays{\sigma_0}$ as an explicit Muller objective~\cite{Horn08} for a zero-sum game played on the arena $G$ from initial vertex $v$. This objective is defined by the set $\{ C \subseteq V \mid \exists \rho \mbox{ such that }\infOcc{\rho}= C \land \rho \in \ObjPlayer{0} \cup \dominatedPlays{\sigma_0} \}$. This exactly encodes the objective of Player~$0$ when he plays the punishing strategy after a deviation $hv$. It is well-known that in zero-sum explicit Muller games, there always exist finite-memory winning strategies with a size exponential in the number $|V|$ of vertices of the arena~\cite{DziembowskiJW97}.
\end{proof}

\begin{lemma}[Reachability] 
\label{lem:reach-case}
Let $v \in V$ and $(w,p) \in \{0,1\} \times \{0,1\}^{\nbrObjectives}$ be such that there exists $hv \in \deviations{\sigma_0}$ with $(\won{hv},\payoff{hv}) = (w,p)$. 

\noindent
Then there exists a finite-memory strategy $\punStrat{(v,w,p)}$ such that for all deviations $hv \in \deviations{\sigma_0}$ with $(\won{hv},\payoff{hv}) = (w,p)$, when Player~$0$ plays $\punStrat{(v,w,p)}$ from $hv$, all consistent plays $\rho$ starting in $v$ are such that either $h\rho \in \ObjPlayer{0}$ or $h\rho \in \dominatedPlays{\sigma_0}$. The size of $\punStrat{(v,w,p)}$ is at most exponential in the size of $\mathcal{G}$.
\end{lemma}
\begin{proof}[Proof of \autoref{lem:reach-case}.]
We follow the same reasoning as in the proof of Lemma~\ref{lem:parity-case}, except that reachability objectives are not prefix-independent. We thus need to take into account the set of objectives $\ObjPlayer{i}$ already satisfied along the history $hv$, which is recorded in $(w,p)$. The uniform finite-memory strategy $\punStrat{(v,w,p)}$ that Player~$0$ can use from all deviations $hv$ such that $\won{hv} = w$ and $\payoff{hv} = p$ is constructed as follows. First, notice that if $w = 1$, meaning that objective $\Omega_0$ is already satisfied, then Player~$0$ can play using any memoryless strategy as punishing strategy. Second, if $w = 0$, as done in \autoref{subsec:fptreach}, we consider the extension of $G$ such that its vertices are of the form $(v', w', p')$ where the $(w', p')$-component keeps track of the objectives that have been satisfied so far and such that its initial vertex is equal to $(v,w,p)$. On this extended arena, we consider the zero-sum game with the objective $\ObjPlayer{0} \cup \dominatedPlays{\sigma_0}$ encoded as the disjunction of a reachability objective ($\ObjPlayer{0}$) and a safety objective ($\dominatedPlays{\sigma_0}$). More precisely, in the extended game, Player~0 has the objective either to reach a vertex in the set $\{ (v', w', p') \mid w' = 1 \}$ or to stay forever within the set of vertices $\{ (v', w', p') \mid \exists p'' \in \paretoSet{\sigma_0} : p' < p'' \}$. It is known, see e.g.~\cite{BruyereHR18}, that there always exist memoryless winning strategies for zero-sum games with an objective which is the disjunction of a reachability objective and a safety objective. Therefore, this is the case here for the extended game, and thus also in the original game however with a winning finite-memory strategy with exponential size. 
\end{proof}

If we systematically change within $\sigma_0$ the behavior of Player~$0$ after a deviation from $\Wit{\sigma_0}$, and use the punishing strategies as defined in the proofs of Lemmas~\ref{lem:parity-case} and~\ref{lem:reach-case}, we obtain a new strategy $\hat{\sigma}_0$ that is solution to the \problemAb{}. The total size of the punishing finite-memory strategies in $\hat{\sigma}_0$ is at most exponential in the size of $\mathcal{G}$. To obtain our results, it remains to show how to compact the plays in $\Wit{\sigma_0}$. To that end, we study the histories and plays within $\Wit{\sigma_0}$. 

\subparagraph*{Compacting Witnesses.}
We now show how to compact the set of witnesses in a way to produce them with a finite-memory strategy. Together with the punishing strategies this will lead to a solution $\tilde{\sigma}_0$ to \problemAb{} with a memory of exponential size. We first consider reachability objectives and explain later how to modify the construction for parity objectives. 

Given a history $h$ that is prefix of at least one witness in $\Wit{\sigma_0}$, we call \emph{\region} and we denote by $\Reg{h}$ the tuple $\Reg{h} = (\won{h},\payoff{h},\prefStrat{h})$.
We also use notation $R = (w,p,\W)$ for a \region. Given a witness $\rho = v_0v_1 \ldots \in \Wit{\sigma_0}$, we consider $\rho^* = (v_0,R_0) (v_1,R_1) \ldots$ such that each $v_j$ is extended with the \region\ $R_j = (w_j,p_j,\W_j) = \Reg{v_0v_1 \dots v_j}$. Similarly we define $h^*$ associated with any history $h$ prefix of a witness. The following properties hold for a witness $\rho$ and its corresponding play $\rho^*$:
\begin{itemize}
    \item for all $j \geq 0$, we have $w_j \leq w_{j+1}$, $p_j \leq p_{j+1}$, and $\W_j \supseteq \W_{j+1}$,
    \item the sequence $(w_j,p_j)_{j\geq 0}$ eventually stabilizes on $(w,p)$ equal to the extended payoff $( \won{\rho}, \payoff{\rho})$ of $\rho$, 
    \item the sequence $(\W_j)_{j\geq 0}$ eventually stabilizes on a set $\W$ which is a singleton such that $\W = \{\rho\}$.
\end{itemize}

Thanks to the previous properties, each $\rho \in \Wit{\sigma_0}$ can be \emph{\region\ decomposed} into a sequence of paths $\pi[1]\pi[2]\cdots\pi[k]$ where the corresponding decomposition $\pi^*[1]\pi^*[2]\cdots\pi^*[k]$ of $\rho^*$ is such that for each $\ell$: \emph{(i)} the \region\ is constant along the path $\pi^*[\ell]$ and \emph{(ii)} it is distinct from the \region\ of the next path $\pi^*[\ell + 1]$ (if $\ell < k$). Each $\pi[\ell]$ is called a \emph{\sect} of $\rho$, such that it is \emph{\internal} (resp.\ \emph{\final}) if $\ell < k$ (resp.\ $\ell = k$).

Notice that the number of regions that are traversed by $\rho$ is bounded by 
\begin{eqnarray} \label{eq:traversed}
(\nbrObjectives + 2) \cdot |\Wit{\sigma_0}|.
\end{eqnarray}
Indeed along $\rho$, the first two components $(w,p)$ of a region correspond to a monotonically increasing vector of $\nbrObjectives + 1$ Boolean values (from $(0,(0,\ldots,0))$ to $(1,(1,\ldots,1))$ in the worst case), and the last component  $W$ is a monotonically decreasing set of witnesses (from $\Wit{\sigma_0}$ to $\{\rho\}$ in the worst case). So the number of regions traversed by a witness is bounded exponentially in the size of the game $\mathcal{G}$.

We have the following important properties for the \sect s of the witnesses of $\Wit{\sigma_0}$.
\begin{itemize}
    \item Let $\rho,\rho' \in \Wit{\sigma_0}$, with region decompositions $\rho = \pi[1]\cdots\pi[k]$ and $\rho' = \pi'[1]\cdots\pi'[k']$ and let $h$ be the longest common prefix of $\rho$ and $\rho'$. Then there exists $k_1 < k,k'$ such that $h = \pi[1]\cdots\pi[k_1]$, $\pi[\ell] = \pi'[\ell]$ for all $\ell \in \{1,\ldots,k_1\}$ and $\pi[k_1+1] \neq \pi'[k_1+1]$. Therefore, when $\Wit{\sigma_0}$ is seen as a tree, the branching structure of this tree is respected by the \sect s. 
    \item Let $R = (w,p,\W)$ be a region and consider the set of all histories $h$ such that $\Reg{h} = R$. Then all these histories are prefixes of each other and are prefixes of exactly $|W|$ witnesses (as $\prefStrat{h} = \W$ for each such $h$). Therefore, the branching structure of $\Wit{\sigma_0}$ is respected by the \sect s such that the associated \region s are all pairwise distinct. The latter property is called the \emph{\Witproperty} of $\Wit{\sigma_0}$.
\end{itemize}
\noindent
We consider a \emph{compact} version $c \Wit{\sigma_0}$ of $\Wit{\sigma_0}$ defined as follows:
\begin{itemize}
    \item each \internal\ \sect\ $\pi$ of $\Wit{\sigma_0}$ is replaced by the elementary path $c \pi$ obtained by eliminating all the cycles of $\pi$. Each \final\ \sect\ $\pi$ of $\Wit{\sigma_0}$ is replaced by a lasso $c \pi = \pi'_1(u\pi'_2)^\omega$ such that $u$ is a vertex, $\pi'_1u\pi'_2$ is an elementary path, and $\pi'_1u\pi'_2u$ is prefix of $\pi$.
    \item each witness $\rho$ of $\Wit{\sigma_0}$ with \region\ decomposition $\rho = \pi[1]\cdots\pi[k]$ is replaced by $c \rho = c \pi[1] \cdots c \pi[k]$ such that each $\pi[\ell]$ is replaced by $c \pi[\ell]$. Notice that as the \region\ is constant inside the \sect s, the \region\ decomposition of $c \rho$ coincide with the sequence of its $c \pi[\ell]$, $\ell \in \{1, \ldots, k\}$.
\end{itemize}
Therefore, by construction of the compact witnesses, the \Witproperty\ of $\Wit{\sigma_0}$ is kept by the set $\{ c\rho \mid \rho \in \Wit{\sigma_0}\}$ and for each $c\rho \in c\Wit{\sigma_0}$,
\begin{eqnarray} \label{eq:payoffkept}
(\won{c \rho},\payoff{c \rho}) =  (\won{\rho},\payoff{\rho}).
\end{eqnarray}
 
We then construct the announced strategy $\tilde{\sigma}_0$ that produces the set $c \Wit{\sigma_0}$ of compact witnesses and after any deviation acts with the adequate punishing strategy (as mentioned in Lemma~\ref{lem:reach-case}). More precisely, let $gv$ be such that $g$ is prefix of a compact witness and $gv$ is not (Player~$1$ deviates from $c \Wit{\sigma_0}$). Then by definition of the compact witnesses, there exists a deviation $hv$ such that $(\won{gv},\payoff{gv}) = (\won{hv},\payoff{hv}) = (w,p)$. Then from $gv$ Player~0 switches to the punishing strategy $\punStrat{(v,w,p)}$.

\begin{lemma}
\label{lem:reach-correct} 
The strategy $\tilde{\sigma}_0$ is a solution to the \problemAb{} for reachability \gamesAb{} and its size is bounded exponentially in the size of the game $\mathcal{G}$.
\end{lemma}

\begin{proof}[Proof of \autoref{lem:reach-correct}.]
Let us first prove that $\tilde{\sigma}_0$ is a solution to the \problemAb{}. \emph{(i)} By (\ref{eq:payoffkept}), the set of extended payoffs of plays in $c \Wit{\sigma_0}$ is equal to the set of extended payoffs of witnesses in $\Wit{\sigma_0}$. This means that with $c \Wit{\sigma_0}$, we keep the same set $\paretoSet{\sigma_0}$ and the objective $\Omega_0$ is satisfied along each compact witness. \emph{(ii)} The punishing strategies used by $\tilde{\sigma}_0$ guarantee the satisfaction of the objective $\ObjPlayer{0} \cup \dominatedPlays{\sigma_0}$ by Lemma~\ref{lem:reach-case}. Therefore $\tilde{\sigma}_0$ is a solution to the \problemAb{}.

Let us now show that the memory size of $\tilde{\sigma}_0$ is bounded exponentially in the size of $\mathcal{G}$. \emph{(i)} By Lemma~\ref{lem:reach-case}, each punishing strategy used by $\tilde{\sigma}_0$ is of exponential size and the number of punishing strategies is exponential. \emph{(ii)} To produce the compact witnesses, $\tilde{\sigma}_0$ keeps in memory the current region and produces in a memoryless way the corresponding compact section (which is an elementary path or lasso). Thus the required memory size for producing $c \Wit{\sigma_0}$ is the number of regions. By (\ref{eq:traversed}), every play in $c \Wit{\sigma_0}$ traverses at most an exponential number of regions and there is an exponential number of such plays (equal to $|P_{\sigma_0} |$).
\end{proof}

We now switch to parity \gamesAb{} and state the following lemma whose proof follows the same line of arguments as those given for reachability objectives.

\begin{lemma}
\label{lem:parity-correct}
The strategy $\tilde{\sigma}_0$ is a solution to the \problemAb{} for parity \gamesAb{} and its size is bounded exponentially in the size of the game $\mathcal{G}$.
\end{lemma}
\begin{proof}[Proof of \autoref{lem:parity-correct}.]
We highlight here the main differences from reachability \gamesAb{}.
 \begin{itemize}
      \item As parity objectives are prefix-independent, we associate to each history $h$ of a play $\rho \in \Wit{\sigma_0}$ a singleton $\Reg{h} = \prefStrat{h}$ instead of the triplet $(\won{h},\payoff{h},\prefStrat{h})$ as in the case of reachability. This is because $(\won{h},\payoff{h})$ does not make sense for prefix-independent objectives.
      \item For the definition of the compact witnesses, we proceed identically as for reachability by simply removing cycles inside each section with the exception of terminal sections. Given the terminal section $\pi[k]$ of a witness $\rho \in \Wit{\sigma_0}$, we replace it by a lasso $c \pi[k] = \pi'_1(\pi'_2)^\omega$ such that $c\pi[k]$ and $\pi[k]$ start at the same vertex, $\occ{c \pi[k]} = \occ{\pi[k]}$, $\infOcc{c \pi[k]} = \infOcc{\pi[k]}$, and $|\pi'_1\pi'_2|$ is quadratic in $|V|$~\cite[Proposition 3.1]{BouyerBMU15}. Therefore, by construction, the objectives $\Omega_i$ satisfied by a witness $\rho$ are exactly the same as for its corresponding compact play $c\rho$. 
  \end{itemize} 

We then construct the strategy $\tilde{\sigma}_0$ that produces the set $c \Wit{\sigma_0}$ of compact witnesses and after any deviation $gv$ from $c \Wit{\sigma_0}$ acts with the adequate punishing strategy $\punStrat{v}$ (as mentioned in Lemma~\ref{lem:parity-case}).
\end{proof}
 
Lemmas~\ref{lem:reach-correct} and~\ref{lem:parity-correct} lead to \autoref{prop:tildesigma}. Using this proposition, we are now able to prove our result on the \nexptime-membership of reachability and parity \gamesAb{}.

\begin{proof}[Proof of Theorem~\ref{thm:nexptime}.] 
We have established the existence of solutions to the \problemAb{} that use a finite memory bounded exponentially, both for reachability (Lemma~\ref{lem:reach-correct}) and for parity (Lemma~\ref{lem:parity-correct}) \gamesAb{}. Let $\sigma_0$ be such a solution. As it is finite-memory, we can guess it as a Moore machine $\mathcal{M}$ with a set of memory states at most exponential in the size of $\mathcal{G}$.

Let us explain how to verify that the guessed solution $\sigma_0$ is a solution to the \problemAb{} for parity objectives, i.e., every play in $\Playsigmazero$ which is \paretoOptimal{} satisfies the objective $\Omega_0$ of Player~$0$. First, we construct the cartesian product $G \times \mathcal{M}$ of the arena $G$ with the Moore machine $\mathcal{M}$ which is a graph whose infinite paths (starting from the initial vertex $v_0$ and the initial memory state) are exactly the plays consistent with $\sigma_0$. Second, to compute $\paretoSet{\sigma_0}$, we test for the existence of a play $\rho$ in $G \times \mathcal{M}$ with a given payoff $p = \payoff{\rho}$, beginning with the largest possible payoff $p = (1,\ldots,1)$ and finishing with the smallest possible one $p = (0,\ldots,0)$. Verifying this corresponds to deciding whether there exists a play that satisfies an intersection of parity objectives. The latter property can be checked in polynomial time in the size of $G \times \mathcal{M}$~\cite{EmersonL87}. Third, to check that each Pareto optimal play in $\Playsigmazero$ satisfies $\Omega_0$, we test for each $p \in \paretoSet{\sigma_0}$ whether there exists a play that satisfies the objectives $\Omega_i$ such that $p_i = 1$ as well as the objective $\Plays_G \setminus \Omega_0$. As the complement of a parity objective is again a parity objective, we use again the polynomial algorithm of~\cite{EmersonL87}. As a consequence we have a \nexptime{} algorithm for parity \gameAb{}s.

The case of reachability \gameAb{}s is solved similarly with the following two differences. Concerning the second step, the existence of a play in $G \times \mathcal{M}$ that satisfies an intersection of reachability objectives can be checked in polynomial time by first extending this graph with a Boolean vector in $\B^\nbrObjectives$ keeping track of the objectives of Player~$1$ already satisfied. Notice that the resulting graph is still of exponential size and that the intersection of reachability objectives becomes a single reachability objective. Concerning the third step, as the complement $\Plays_G \setminus \Omega_0$ of $\Omega_0$ is not a reachability objective, we rather remove vertices of $G \times \mathcal{M}$ that contains an element of the target set $T_0$ before checking whether there exists a play that satisfies the objectives $\Omega_i$ such that $p_i = 1$ for a given $p \in \paretoSet{\sigma_0}$.
\end{proof}

\subsection{\textsf{NP}-Completeness for Tree Arenas}
Before turning to the \nexptime-hardness of the \problemAb{} in the next section, we first want to show that the \problemAb{} is already \npComplete{} in the simple setting of reachability objectives and arenas that are trees. To do so, we use a reduction from the \setCover{} (\setCoverAb{}) which is \npComplete{}~\cite{Karp72}.
 
\begin{theorem} \label{thm:npcomplete}
    The \problemAb{} is \npComplete{} for reachability \gamesAb{} on tree arenas. 
\end{theorem}

Notice that when the game arena is a tree, it is easy to design an algorithm for solving the \problemAb{} that is in \np{}. First, we nondeterministically guess a strategy $\sigma_0$ that can be assumed to be memoryless as the arena is a tree. Second, we apply a depth-first search algorithm from the root vertex which accumulates to leaf vertices the extended payoff of plays which are consistent with $\sigma_0$. Finally, we check that $\sigma_0$ is a solution.

Let us explain why the \problemAb{} is \npHard{} on tree arenas by reduction from the \setCoverAb{}. We recall that an instance of the \setCoverAb{} is defined by a set $C = \{e_1,e_2, \dots, e_n\}$ of $n$ elements, $m$ subsets $S_1, S_2, \dots, S_m$ such that $S_i \subseteq C$ for each $i\in \{1, \dots, m\}$, and an integer $\problemParam \leq m$. The problem consists in finding $\problemParam$ indexes $i_1, i_2, \dots, i_\problemParam$ such that the union of the corresponding subsets equals $C$, i.e., $C = \bigcup\limits_{j = 1}^\problemParam S_{i_j}$.

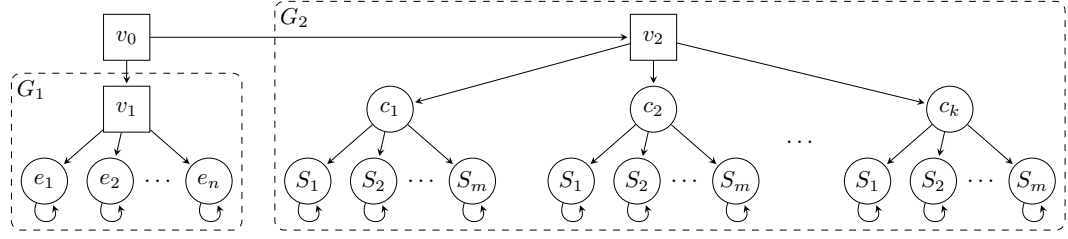
\begin{figure}
	\centering
	\resizebox{\textwidth}{!}{%
		
		\begin{tikzpicture}[->,>=stealth, shorten >=1pt,auto]
		
		\draw[dashed, rounded corners] (-1.5,0.15) rectangle (2, -2.25) {};
		\node[] at (-1.2,-0.1) {$G_1$};
		
		\node[draw, rectangle, minimum size=2em,inner sep=1] (elements) at (0.25,-0.4){$v_1$};
		\node[draw, circle, minimum size=2em,inner sep=1] (1) at (-1,-1.5){$e_1$};
		\node[draw, circle, minimum size=2em,inner sep=1] (2) at (0,-1.5){$e_2$};
		\node[] (dots) at (0.75,-1.5){$\dots$};
		\node[draw, circle, minimum size=2em,inner sep=1] (m) at (1.5,-1.5){$e_n$};
		
		\draw[-stealth] (elements) edge [] node {} (1);
		\draw[-stealth] (elements) edge [] node {} (2);
		\draw[-stealth] (elements) edge [] node {} (m);

		\draw[-stealth,in=290, out=250, looseness=4] (1) edge [ ] node {} (1);
		\draw[-stealth,in=290, out=250, looseness=4] (2) edge [ ] node {} (2);
		\draw[-stealth,in=290, out=250, looseness=4] (m) edge [ ] node {} (m);
		
        \draw[dashed, rounded corners] (2.5,1.25) rectangle (14.5, -2.25) {};
		\node[] at (2.8,1) {$G_2$};
        
		\node[draw, circle, minimum size=2em,inner sep=1] (choice_1) at (4.25,-0.4){$c_1$};
		\node[draw, circle, minimum size=2em,inner sep=1] (set_1) at (3,-1.5){$S_1$};
		\node[draw, circle, minimum size=2em,inner sep=1] (set_2) at (4,-1.5){$S_2$};
		\node[] (next_dots) at (4.75,-1.5){$\dots$};
		\node[draw, circle, minimum size=2em,inner sep=1] (set_m) at (5.5,-1.5){$S_m$};
		
		\draw[-stealth] (choice_1) edge [] node {} (set_1);
		\draw[-stealth] (choice_1) edge [] node {} (set_2);
		\draw[-stealth] (choice_1) edge [] node {} (set_m);
		\draw[-stealth,in=290, out=250, looseness=4] (set_1) edge [ ] node {} (set_1);
		\draw[-stealth,in=290, out=250, looseness=4] (set_2) edge [ ] node {}  (set_2);
		\draw[-stealth,in=290, out=250, looseness=4] (set_m) edge [ ] node {} (set_m);


		\node[draw, circle, minimum size=2em,inner sep=1] (choice_2) at (8.25,-0.4){$c_2$};
		\node[draw, circle, minimum size=2em,inner sep=1] (set_12) at (7,-1.5){$S_1$};
		\node[draw, circle, minimum size=2em,inner sep=1] (set_22) at (8,-1.5){$S_2$};
		\node[] (next_dots2) at (8.75,-1.5){$\dots$};
		\node[draw, circle, minimum size=2em,inner sep=1] (set_m2) at (9.5,-1.5){$S_m$};
		
		\draw[-stealth] (choice_2) edge [] node {} (set_12);
		\draw[-stealth] (choice_2) edge [] node {} (set_22);
		\draw[-stealth] (choice_2) edge [] node {} (set_m2);
		\draw[-stealth,in=290, out=250, looseness=4] (set_12) edge [ ] node {} (set_12);
		\draw[-stealth,in=290, out=250, looseness=4] (set_22) edge [ ] node {}  (set_22);
		\draw[-stealth,in=290, out=250, looseness=4] (set_m2) edge [ ] node {} (set_m2);
		
		\node[minimum size=0.8cm] (next_dots_inv) at (10.5,-0.9){$\dots$};


		\node[draw, circle, minimum size=2em,inner sep=1] (choice_t) at (12.75,-0.4){$c_\problemParam$};
		\node[draw, circle, minimum size=2em,inner sep=1] (set_1t) at (11.5,-1.5){$S_1$};
		\node[draw, circle, minimum size=2em,inner sep=1] (set_2t) at (12.5,-1.5){$S_2$};
		\node[] (next_dotst) at (13.25,-1.5){$\dots$};
		\node[draw, circle, minimum size=2em,inner sep=1] (set_mt) at (14,-1.5){$S_m$};
		
		\draw[-stealth] (choice_t) edge [] node {} (set_1t);
		\draw[-stealth] (choice_t) edge [] node {} (set_2t);
		\draw[-stealth] (choice_t) edge [] node {} (set_mt);
		\draw[-stealth,in=290, out=250, looseness=4] (set_1t) edge [] node {} (set_1t);
		\draw[-stealth,in=290, out=250, looseness=4] (set_2t) edge [] node {}  (set_2t);
		\draw[-stealth,in=290, out=250, looseness=4] (set_mt) edge [] node {} (set_mt);
		
        \node[draw, rectangle, minimum size=2em,inner sep=1] (root) at (0.25,0.7){$v_0$};
        \node[draw, rectangle, minimum size=2em,inner sep=1] (0) at (8.25,0.7){$v_2$};
		\draw[-stealth] (0) edge [] node {} (choice_1);
		\draw[-stealth] (0) edge [] node {} (choice_2);
		\draw[-stealth] (0) edge [] node {} (choice_t);
		
		\draw[-stealth] (root) edge [] node {} (elements);
		\draw[-stealth] (root) edge [] node {} (0);

		\end{tikzpicture}
		}%
	
	\caption{The tree arena used in the reduction from the \setCoverAb{}.}
	\label{scp_game}
\end{figure}

Given an instance of the \setCoverAb{}, we construct a game with an arena consisting of $n + \problemParam \cdot (m + 1) + 3$ vertices. The arena $G$ of the game is provided in Figure \ref{scp_game} and can be seen as two sub-arenas reachable from the initial vertex $v_0$. The game is such that there is a solution to the \setCoverAb{} if and only if Player~$0$ has a strategy from $v_0$ in $G$ which is a solution to the \problemAb{}. The game is played between Player~$0$ with reachability objective $\Omega_0$ and Player~$1$ with $n + 1$ reachability objectives. The objectives are defined as follows: $\ObjPlayer{0} = \reach{\{v_2\}}$, $\ObjPlayer{i} = \reach{\{e_i\} \cup \{S_j \mid e_i \in S_j\}}$ for $i \in \{1,2, \dots, n\}$ and $\ObjPlayer{n+1} = \reach{\{v_2\}}$. First, notice that every play in $G_1$ is consistent with any strategy of Player~$0$ and is lost by that player. It holds that for each $\ell \in \{1, 2, \dots, n\}$, there is such a play with payoff $(p_1, \ldots, p_{n+1})$ such that $p_\ell = 1$ and $p_j = 0$ for $j \neq \ell$. These payoffs correspond to the elements $e_\ell$ we aim to cover in the \setCoverAb{}. A play in $G_2$ visits $v_2$ and then a vertex $c$ from which Player~$0$ selects a vertex $S$. Such a play is always won by Player~$0$ and its payoff is $(p_1, \ldots, p_{n+1})$ such that $p_{n+1} = 1$ and $p_r = 1$ if and only if the element $e_r$ belongs to the set $S$. It follows that the payoff of such a play corresponds to a set of elements in the \setCoverAb{}. It is easy to see that the following proposition holds and it follows that, as a consequence, Theorem~\ref{thm:npcomplete} holds.

\begin{restatable}{proposition}{proptreenphard}
\label{prop:tree_nphard}
    There is a solution to an instance of the \setCoverAb{} if and only if Player~$0$ has a strategy from $v_0$ in the corresponding \gameAb{} that is a solution to the \problemAb{}.
\end{restatable}

\begin{proof}[Proof of \autoref{prop:tree_nphard}.]
    First, let us assume that there is a solution to the \setCoverAb{}. It holds that there exists a set of $\problemParam$ indexes $i_1, i_2, \dots, i_\problemParam$ such that the union of the corresponding sets equals the set $C$ of elements we aim to cover. We define the strategy $\sigma_0$ as follows: $\sigma_0(v_0 v_2 c_j) = S_{i_j}$. Let us show that this strategy is solution to the \problemAb{} by showing that any play with a \paretoOptimal{} payoff is won by Player~$0$. This amounts to showing that for every play in $G_1$ there is a play in $G_2$ with a strictly larger payoff. This is sufficient as it makes sure that the payoff of plays in $G_1$ are not \paretoOptimal{} and as every play in $G_2$ is won by Player~$0$. Let $p =(p_1, \ldots, p_{n+1})$ be the payoff of a play in $G_1$. It holds that $p_\ell = 1$ for some $\ell \in \{1, 2, \dots, n\}$ and $p_j = 0$ for $\ell \neq j$. This corresponds to the element $e_\ell$ in $C$. Since the $\problemParam$ indexes $i_1, i_2, \dots, i_\problemParam$ are a solution to the \setCoverAb{}, it holds that there exists some index $i_j$ such that $e_\ell \in S_{i_j}$. It also holds that the play $v_0 v_2 c_j (S_{i_j})^\omega$ is consistent with $\sigma_0$. Its payoff is $p' = (p'_1, \ldots, p'_{n+1})$ with $p'_\ell = 1$ since $e_\ell \in S_{i_j}$ and $p'_{n+1} = 1$. It follows that payoff $p'$ is strictly larger than $p$. 
    
    Now, let us assume that Player~$0$ has a strategy $\sigma_0$ from $v_0$ that is a solution to the \problemAb{}. We can show that the set of indexes $\{i_j \mid \sigma_0(v_0 v_2 c_j) = S_{i_j}, j \in \{1, \dots, \problemParam \}\}$ is a solution to the \setCoverAb{}. It is easy to see that since strategy $\sigma_0$ is a solution to the \problemAb{}, every payoff $p$ in $G_1$ is strictly smaller than some payoff $p'$ in $G_2$. It follows that in the \setCoverAb{}, each element $e \in C$ corresponding to $p$ is contained in some set $S$ corresponding to $p'$. Since it also holds that $S \subseteq C$ for each set $S$, it follows that the sets mentioned above are an exact cover of $C$.
\end{proof}

\subsection{\textsf{NEXPTIME}-Hardness}

Let us come back to regular game arenas and show the \nexptime{}-hardness result for both reachability and parity \gamesAb{}. Each type of objective is studied in a dedicated subsection. 

\begin{restatable}{theorem}{thmnexptimehardreach}
\label{thm:nexptimehard-reach}
    The \problemAb{} is \nexptimeHard{} for reachability \gamesAb{}.  
\end{restatable}
\begin{restatable}{theorem}{thmnexptimehardparity}
\label{thm:nexptimehard-parity}
    The \problemAb{} is \nexptimeHard{} for parity \gamesAb{}.  
\end{restatable}

The \nexptime{}-hardness is obtained thanks to the succinct variant of the \setCoverAb{} presented below.

\subsubsection{Succinct Set Cover Problem}

The \emph{\succinctSetCover{} (\succinctSetCoverAb{})} is defined as follows. We are given a Conjunctive Normal Form (CNF) formula $\phi = C_1 \land C_2 \land \dots \land C_p$ over the variables $X = \{x_1, x_2, \ldots, x_m\}$ made up of $p$ clauses, each containing some disjunction of literals of the variables in $X$. The set of valuations of the variables $X$ which satisfy $\phi$ is written $\llbracket \phi \rrbracket$. We are also given an integer $\problemParam \in \mathbb{N}$ (encoded in binary) and an other CNF formula $\psi = D_1 \land D_2 \land \dots \land D_q$ over the variables $X \cup Y$ with $ Y = \{y_1, y_2, \ldots, y_n \}$, made up of $q$ clauses. Given a valuation $val_Y: Y \rightarrow \{0, 1\}$ of the variables in $Y$, called a \emph{partial valuation}, we write $\psi[val_Y]$ the CNF formula obtained by replacing in $\psi$ each variable $y \in Y$ by its valuation $val_Y(y)$. We write $\llbracket \psi[val_Y] \rrbracket$ the valuations of the remaining variables $X$ which satisfy $\psi[val_Y]$. The \succinctSetCoverAb{} is to decide whether there exists a set $K = \big{\{}val_Y \mid val_Y: Y \rightarrow \{0, 1\} \big{\}}$ of $\problemParam$ valuations of the variables in $Y$ such that the valuations of the remaining variables $X$ which satisfy the formulas $\psi[val_Y]$ include the valuations of $X$ which satisfy $\phi$. Formally, we write this $\llbracket \phi \rrbracket \subseteq \bigcup\limits_{val_Y \in K} \llbracket \psi[val_Y] \rrbracket$. 

We can show that this corresponds to a set cover problem succinctly defined using CNF formulas. The set $\llbracket \phi \rrbracket$ of valuations of $X$ which satisfy $\phi$ corresponds to the set of elements we aim to cover. Parameter $\problemParam$ is the number of sets that can be used to cover these elements. Such a set is described by a formula $\psi[val_Y]$, given a partial valuation $val_Y$, and its elements are the valuations of $X$ in $\llbracket \psi[val_Y] \rrbracket$. This is illustrated in the following example.

\begin{example}
\label{example_phi}
Consider the CNF formula $\phi = (x_1 \lor \neg x_2) \land (x_2 \lor x_3)$ over the variables $X = \{x_1, x_2, x_3\}$. The set of valuations of the variables which satisfy $\phi$ is $\llbracket \phi \rrbracket = \{(1,1,1), (1,1,0), (1,0,1), (0,0,1)\}$. Each such valuation corresponds to one element we aim to cover. Consider the CNF formula $\psi = (y_1 \lor y_2) \land (x_1 \lor y_2) \land (x_2 \lor x_3 \lor y_1)$ over the variables $X \cup Y$ with $Y = \{y_1, y_2\}$. Given the partial valuation $val_Y$ of the variables in $Y$ such that $val_Y(y_1) = 0$ and $val_Y(y_2) = 1$, we get the CNF formula $\psi[val_Y] = (0 \lor 1) \land (x_1 \lor 1) \land (x_2 \lor x_3 \lor 0)$. This formula describes the contents of the set identified by the partial valuation (as a partial valuation yields a unique formula). The valuations of the variables $X$ which satisfy $\psi[val_Y]$ are the elements contained in the set. In this case, these elements are $\llbracket \psi[val_Y] \rrbracket = \{(0,1,0), (0,0,1), (0,1,1), (1,1,0), (1,0,1), (1,1,1)\}$. We can see that this set contains the elements $\{(1,1,1), (1,1,0), (1,0,1), (0,0,1)\}$ of $\llbracket \phi \rrbracket$.
\end{example}

The following result is used in the proof of our \nexptime{}-hardness results and is of potential independent interest.

\begin{restatable}{theorem}{thmssccompleteness}
\label{thm:ssc-completeness}
The \succinctSetCoverAb{} is \nexptimeComplete{}.
\end{restatable}

\begin{proof}[Proof of \autoref{thm:ssc-completeness}.] 
It is easy to see that the \succinctSetCoverAb{} is in \nexptime{}. We can show that the \succinctSetCoverAb{} is \nexptimeHard{} by reduction from the \emph{\dominatingSet{}} (\dominatingSetAb{}) which is known to be \nexptimeComplete{} for graphs \emph{succinctly} defined using CNF formulas \cite{DasST17}. An instance of the \dominatingSetAb{} is defined by a CNF formula $\theta$ over two sets of $n$ variables $X = \{x_1, x_2, \dots, x_n\}$ and $Y=\{y_1, y_2,\dots, y_n\}$ and an integer $\problemParam$ (encoded in binary). The formula $\theta$ succinctly defines an undirected graph in the following way. The set of vertices is the set of all valuations of the $n$ variables in $X$ (or over the $n$ variables in $Y$) of which there are $2^n$. Let $val_X$ and $val_Y$ be two such valuations, representing two vertices. Then, there is an edge between $val_X$ and $val_Y$ if and only if $\theta[val_X, val_Y]$ or $\theta[val_Y, val_X]$ is true. An instance of the \dominatingSetAb{} is positive if there exists a set $K = \{val^1_X, val^2_X, \dots, val^\problemParam_X\}$ of $\problemParam$ valuations of the variables in $X$, corresponding to $\problemParam$ vertices, such that all vertices in the graph are adjacent to a vertex in $K$. Formally, we write this $|\bigcup\limits_{val_X \in K} \{ val_Y \mid \theta[val_X, val_Y] \lor  \theta[val_Y, val_X] \mbox{ is true} \}| = 2^n$.

The \dominatingSetAb{} can be reduced in polynomial time to the \succinctSetCoverAb{} as follows. We define the CNF formula $\phi$ over the set of variables $X$ such that the formula is empty. Therefore, the set $\llbracket \phi\rrbracket$ is equal to the $2^n$ valuations of the variables in $X$. We then define the CNF formula $\psi$ over the set of variables $X$ and $Y$ such that it is the CNF equivalent to $\theta(X, Y) \lor \theta(Y, X)$.  
The latter formula has a size which is polynomial in the size of the CNF formula $\theta$ which defines the graph. We keep the same integer $\problemParam$. Then, it is direct to see that the instance of the \dominatingSetAb{} is positive if and only if the instance of \succinctSetCoverAb{} is positive. Indeed, there is a positive instance to the \dominatingSetAb{} if and only if there exists a set $K$ of $\problemParam$ valuations of the variables in $Y$ such that $\llbracket \phi \rrbracket \subseteq \bigcup\limits_{val_Y \in K} \llbracket \psi[val_Y] \rrbracket$.
\end{proof}

\subsubsection{\textsf{NEXPTIME}-Hardness of Reachability SP Games}

We now describe in details our reduction from the \succinctSetCoverAb{} which allows us to show the \nexptime-hardness of solving the \problemAb{} in reachability \gamesAb{}.\\

Given an instance of the \succinctSetCoverAb{}, we construct a reachability \gameAb{} with arena $G$ consisting of a polynomial number of vertices in the number of clauses and variables in the formulas $\phi$ and $\psi$ and in the length of the binary encoding of the integer $k$. This reduction is such that there is a solution to the \succinctSetCoverAb{} if and only if Player~$0$ has a strategy from $v_0$ in $G$ which is a solution to the \problemAb{}. The arena $G$, provided in Figure \ref{sscp_game_long}, can be viewed as three sub-arenas reachable from $v_0$. We call these sub-arenas $G_1$, $G_2$ and $G_3$. Sub-arena $G_3$ starts with a gadget $Q_\problemParam$ whose vertices belong to Player~$1$ and which provides exactly $k$ different paths from $v_0$ to $v_3$.

\subparagraph*{Gadget $Q_\problemParam$.} 
Parameter $\problemParam$ can be represented in binary using $r = \lfloor log_2(\problemParam) \rfloor + 1$ bits. It also holds that the binary encoding of $\problemParam$ corresponds to the sum of at most $r$ powers of 2. Given the binary encoding $b_0 b_1 \dots b_{r-1}$ of $\problemParam$ such that $b_i \in \B$, let $ones = \{ i \in \{0, \dots, r-1\} \mid b_i = 1\}$. It holds that $\problemParam = \sum_{i \in ones}^{} 2^i$. Our gadget $Q_\problemParam$ is a graph with a polynomial number of vertices (in the length of the binary encoding of $k$) such that all these vertices belong to Player~$1$. For each $i \in ones$ there is $2^i$ different paths from the initial vertex $\alpha$ to vertex $\beta$. Therefore, it holds that in $Q_\problemParam$ there are $\problemParam$ different paths from vertex $\alpha$ to vertex $\beta$. 
\begin{example}
Let $\problemParam = 11$, it holds that it can be represented in binary using $\lfloor log_2(11) \rfloor + 1 = 4$ bits. The binary representation of $11$ is $1011$ and it can be obtained by the following sum $2^3 + 2^1 + 2^0$. The gadget $Q_{11}$ is detailed in Figure \ref{gadget}.

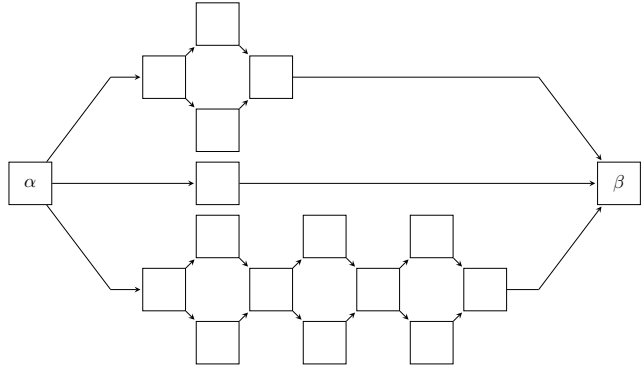
\begin{figure}
	\centering
	\resizebox{0.6\textwidth}{!}{%
		
		\begin{tikzpicture}
		
		\node[draw, rectangle, minimum size=0.8cm] (al) at (-1,0.5){$\alpha$};
		
		\node[draw, rectangle, minimum size=0.8cm] (t1) at (2.5,0.5){};

		\node[draw, rectangle, minimum size=0.8cm] (m1) at (1.5,2.5){};
		\node[draw, rectangle, minimum size=0.8cm] (m11) at (2.5,1.5){};
		\node[draw, rectangle, minimum size=0.8cm] (m12) at (2.5,3.5){};
		\node[draw, rectangle, minimum size=0.8cm] (m2) at (3.5,2.5){};
		
		\node[draw, rectangle, minimum size=0.8cm] (b1) at (1.5,-1.5){};
		\node[draw, rectangle, minimum size=0.8cm] (b11) at (2.5,-0.5){};
		\node[draw, rectangle, minimum size=0.8cm] (b12) at (2.5,-2.5){};

		\node[draw, rectangle, minimum size=0.8cm] (b2) at (3.5,-1.5){};
		\node[draw, rectangle, minimum size=0.8cm] (b21) at (4.5,-0.5){};
		\node[draw, rectangle, minimum size=0.8cm] (b22) at (4.5,-2.5){};
		
		\node[draw, rectangle, minimum size=0.8cm] (b3) at (5.5,-1.5){};
		\node[draw, rectangle, minimum size=0.8cm] (b31) at (6.5,-0.5){};
		\node[draw, rectangle, minimum size=0.8cm] (b32) at (6.5,-2.5){};
		
		\node[draw, rectangle, minimum size=0.8cm] (b4) at (7.5,-1.5){};
		
	    \node[draw, rectangle, minimum size=0.8cm] (be) at (10,0.5){$\beta$};

		\draw[-stealth, shorten >=1pt,auto] (al) to [] node []{} (t1);

		\draw[] (al) to [] (0.5, 2.5);
		\draw[] (al) to [] (0.5, -1.5);
		\draw[-stealth, shorten >=1pt,auto] (0.5, 2.5) to [] node []{} (m1);
		\draw[-stealth, shorten >=1pt,auto] (0.5, -1.5) to [] node []{} (b1);

		\draw[-stealth, shorten >=1pt,auto] (m1) to [] node []{} (m11);
		\draw[-stealth, shorten >=1pt,auto] (m1) to [] node []{} (m12);
		\draw[-stealth, shorten >=1pt,auto] (m11) to [] node []{} (m2);
		\draw[-stealth, shorten >=1pt,auto] (m12) to [] node []{} (m2);

		\draw[-stealth, shorten >=1pt,auto] (b1) to [] node []{} (b11);
		\draw[-stealth, shorten >=1pt,auto] (b1) to [] node []{} (b12);
		\draw[-stealth, shorten >=1pt,auto] (b11) to [] node []{} (b2);
		\draw[-stealth, shorten >=1pt,auto] (b12) to [] node []{} (b2);
		
		\draw[-stealth, shorten >=1pt,auto] (b2) to [] node []{} (b21);
		\draw[-stealth, shorten >=1pt,auto] (b2) to [] node []{} (b22);
		\draw[-stealth, shorten >=1pt,auto] (b21) to [] node []{} (b3);
		\draw[-stealth, shorten >=1pt,auto] (b22) to [] node []{} (b3);
		
		\draw[-stealth, shorten >=1pt,auto] (b3) to [] node []{} (b31);
		\draw[-stealth, shorten >=1pt,auto] (b3) to [] node []{} (b32);
		\draw[-stealth, shorten >=1pt,auto] (b31) to [] node []{} (b4);
		\draw[-stealth, shorten >=1pt,auto] (b32) to [] node []{} (b4);
		
		\draw[] (m2) to [] (8.5, 2.5);
		\draw[-stealth, shorten >=1pt,auto] (8.5, 2.5) to [] node []{} (be);
		
		\draw[-stealth, shorten >=1pt,auto] (t1) to [] node []{} (be);

		\draw[] (b4) to [] (8.5, -1.5);
		\draw[-stealth, shorten >=1pt,auto] (8.5, -1.5) to [] node []{} (be);
		
		\end{tikzpicture}
		}%
	
	\caption{The gadget $Q_{11}$.}
	\label{gadget}
\end{figure}

\end{example}

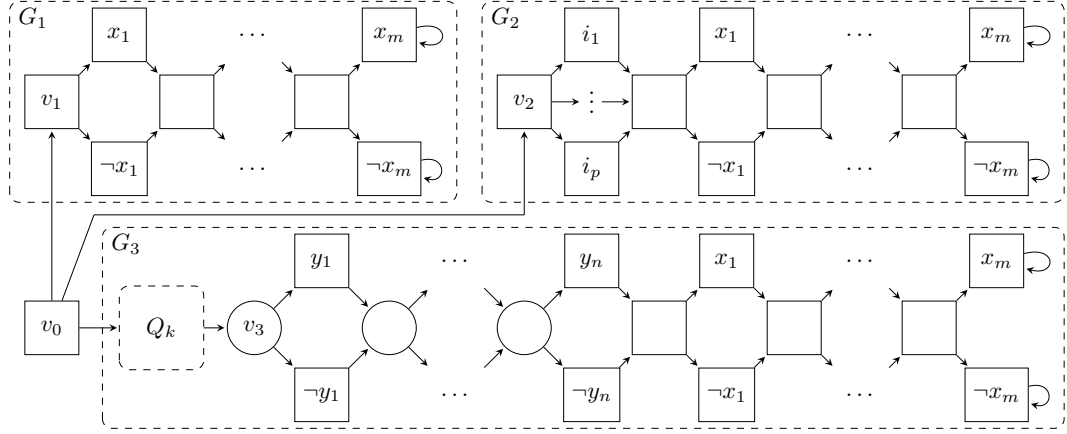
\begin{figure}
	\centering
	\resizebox{\textwidth}{!}{%
		
		\begin{tikzpicture}
		
		\node[draw, rectangle, minimum size=0.8cm] (o) at (-5,-4){$v_0$};
		
		
		\node[draw, rectangle, minimum size=0.8cm] (i0) at (2,-0.625){$v_2$};

		\draw[dashed, rounded corners] (1.375,-2.125) rectangle (10, 0.875)  {};
        \node[] at (1.725,0.625) {$G_2$};
		        
		\node[draw, rectangle, minimum size=0.8cm] (i1) at (3,0.375){$i_1$};
		\node[] (idot) at (3,-0.525){$\vdots$};
		\node[draw, rectangle, minimum size=0.8cm] (in) at (3,-1.625){$i_p$};
		
		\node[draw, rectangle, minimum size=0.8cm] (b1) at (4,-0.625){};
		\node[draw, rectangle, minimum size=0.8cm] (x11) at (5,0.375){$x_1$};
		\node[draw, rectangle, minimum size=0.8cm] (nx11) at (5,-1.625){$\neg x_1$};
		\node[draw, rectangle, minimum size=0.8cm] (b2) at (6,-0.625){};
		\node[minimum size=0.8cm] (b2i1) at (7,0.375){$\dots$};
		\node[minimum size=0.8cm] (b2i2) at (7,-1.625){$\dots$};
		\node[draw, rectangle, minimum size=0.8cm] (bm) at (8,-0.625){};
		\node[draw, rectangle, minimum size=0.8cm] (x1m) at (9,0.375){$x_m$};
		\node[draw, rectangle, minimum size=0.8cm] (nx1m) at (9,-1.625){$\neg x_m$};
		
		\draw[-] (o) to [] (-4.375,-2.3125);
		\draw[-] (-4.375,-2.3125) to [] node []{} (2,-2.3125);
		\draw[-stealth, shorten >=1pt,auto] (2,-2.3125) to [] node []{} (i0);
		\draw[-stealth, shorten >=1pt,auto] (i0) to [] node []{} (i1);
		\draw[-stealth, shorten >=1pt,auto] (i0) to [] node []{} (2.85,-0.625);
		\draw[-stealth, shorten >=1pt,auto] (i0) to [] node []{} (in);
		\draw[-stealth, shorten >=1pt,auto] (i1) to [] node []{} (b1);
		\draw[-stealth, shorten >=1pt,auto] (3.15,-0.625) to [] node []{} (b1);
		\draw[-stealth, shorten >=1pt,auto] (in) to [] node []{} (b1);
		\draw[-stealth, shorten >=1pt,auto] (b1) to [] node []{} (x11);
		\draw[-stealth, shorten >=1pt,auto] (b1) to [] node []{} (nx11);
		\draw[-stealth, shorten >=1pt,auto] (x11) to [] node []{} (b2);
		\draw[-stealth, shorten >=1pt,auto] (nx11) to [] node []{} (b2);
		\draw[-stealth, shorten >=1pt,auto] (b2) to [] node []{} (b2i1);
		\draw[-stealth, shorten >=1pt,auto] (b2) to [] node []{} (b2i2);
		\draw[-stealth, shorten >=1pt,auto] (b2i1) to [] node []{} (bm);
		\draw[-stealth, shorten >=1pt,auto] (b2i2) to [] node []{} (bm);
		\draw[-stealth, shorten >=1pt,auto] (bm) to [] node []{} (x1m);
		\draw[-stealth, shorten >=1pt,auto] (bm) to [] node []{} (nx1m);		
		\draw[-stealth,shorten >=1pt,auto,in=-15,out=15,looseness=6] (x1m) edge [] node {} (x1m);
		\draw[-stealth,shorten >=1pt,auto,in=-15,out=15,looseness=4] (nx1m) edge [] node {} (nx1m);

		\draw[dashed, rounded corners] (-5.625, -2.125) rectangle (1, 0.875) {};
		\node[] at (-5.275,0.625) {$G_1$};
		
		\node[draw, rectangle, minimum size=0.8cm] (a1) at (-5,-0.625){$v_1$};
		\node[draw, rectangle, minimum size=0.8cm] (x21) at (-4,0.375){$x_1$};
		\node[draw, rectangle, minimum size=0.8cm] (nx21) at (-4,-1.625){$\neg x_1$};
	    \node[draw, rectangle, minimum size=0.8cm] (a2) at (-3,-0.625){};
		\node[minimum size=0.8cm] (a2i1) at (-2,0.375){$\dots$};
		\node[minimum size=0.8cm] (a2i2) at (-2,-1.625){$\dots$};
		\node[draw, rectangle, minimum size=0.8cm] (am) at (-1,-0.625){};
		\node[draw, rectangle, minimum size=0.8cm] (x2m) at (0,0.375){$x_m$};
		\node[draw, rectangle, minimum size=0.8cm] (nx2m) at (0,-1.625){$\neg x_m$};
		
		\draw[-stealth, shorten >=1pt,auto] (o) to [] node []{} (a1);
		\draw[-stealth, shorten >=1pt,auto] (a1) to [] node []{} (x21);
		\draw[-stealth, shorten >=1pt,auto] (a1) to [] node []{} (nx21);
		\draw[-stealth, shorten >=1pt,auto] (x21) to [] node []{} (a2);
		\draw[-stealth, shorten >=1pt,auto] (nx21) to [] node []{} (a2);
		\draw[-stealth, shorten >=1pt,auto] (a2) to [] node []{} (a2i1);
		\draw[-stealth, shorten >=1pt,auto] (a2) to [] node []{} (a2i2);
		\draw[-stealth, shorten >=1pt,auto] (a2i1) to [] node []{} (am);
		\draw[-stealth, shorten >=1pt,auto] (a2i2) to [] node []{} (am);
		\draw[-stealth, shorten >=1pt,auto] (am) to [] node []{} (x2m);
		\draw[-stealth, shorten >=1pt,auto] (am) to [] node []{} (nx2m);		
		\draw[-stealth,shorten >=1pt,auto,in=-15,out=15,looseness=6] (x2m) edge [] node {} (x2m);
		\draw[-stealth,shorten >=1pt,auto,in=-15,out=15,looseness=4] (nx2m) edge [] node {} (nx2m);
		
		
		\draw[dashed, rounded corners] (-4.25,-5.5) rectangle (10, -2.5) {};
		\node[] at (-3.9,-2.75) {$G_3$};

		\draw[dashed, rounded corners] (-4,-4.625) rectangle (-2.75, -3.375) {};
		\node[rectangle, minimum size=0.8cm] (j1) at (-3.375,-4){$Q_\problemParam$};

		\draw[-stealth, shorten >=1pt,auto] (o) to [] ((-4,-4);
		
		
		\node[draw, circle, minimum size=0.8cm] (c1) at (-2,-4){$v_3$};
		\node[draw, rectangle, minimum size=0.8cm] (y1) at (-1,-5){$\neg y_1$};
		\node[draw, rectangle, minimum size=0.8cm] (ny1) at (-1,-3){$y_1$};
		\node[draw, circle, minimum size=0.8cm] (c2) at (0,-4){};
		\node[minimum size=0.8cm] (ci1) at (1,-3){$\dots$};
		\node[minimum size=0.8cm] (ci2) at (1,-5){$\dots$};
		\node[draw, circle, minimum size=0.8cm] (cm) at (2,-4){};
		\node[draw, rectangle, minimum size=0.8cm] (ym) at (3,-5){$\neg y_n$};
		\node[draw, rectangle, minimum size=0.8cm] (nym) at (3,-3){$y_n$};
		
		\draw[-stealth, shorten >=1pt,auto] ((-2.75,-4) to [] node []{} (c1);

		\draw[-stealth, shorten >=1pt,auto] (c1) to [] node []{} (y1);
		\draw[-stealth, shorten >=1pt,auto] (c1) to [] node []{} (ny1);
		\draw[-stealth, shorten >=1pt,auto] (y1) to [] node []{} (c2);
		\draw[-stealth, shorten >=1pt,auto] (ny1) to [] node []{} (c2);
		\draw[-stealth, shorten >=1pt,auto] (c2) to [] node []{} (ci1);
		\draw[-stealth, shorten >=1pt,auto] (c2) to [] node []{} (ci2);
		\draw[-stealth, shorten >=1pt,auto] (ci1) to [] node []{} (cm);
		\draw[-stealth, shorten >=1pt,auto] (ci2) to [] node []{} (cm);
		\draw[-stealth, shorten >=1pt,auto] (cm) to [] node []{} (ym);
		\draw[-stealth, shorten >=1pt,auto] (cm) to [] node []{} (nym);			
		
		
		\node[draw, rectangle, minimum size=0.8cm] (d1) at (4,-4){};
		\node[draw, rectangle, minimum size=0.8cm] (x31) at (5,-5){$\neg x_1$};
		\node[draw, rectangle, minimum size=0.8cm] (nx31) at (5,-3){$x_1$};
		\node[draw, rectangle, minimum size=0.8cm] (d2) at (6,-4){};
		\node[minimum size=0.8cm] (d2i1) at (7,-5){$\dots$};
		\node[minimum size=0.8cm] (d2i2) at (7,-3){$\dots$};
		\node[draw, rectangle, minimum size=0.8cm] (dm) at (8,-4){};
		\node[draw, rectangle, minimum size=0.8cm] (x3m) at (9,-5){$\neg x_m$};
		\node[draw, rectangle, minimum size=0.8cm] (nx3m) at (9,-3){$x_m$};
		
		\draw[-stealth, shorten >=1pt,auto] (ym) to [] node []{} (d1);		
		\draw[-stealth, shorten >=1pt,auto] (nym) to [] node []{} (d1);	
		\draw[-stealth, shorten >=1pt,auto] (d1) to [] node []{} (x31);
		\draw[-stealth, shorten >=1pt,auto] (d1) to [] node []{} (nx31);
		\draw[-stealth, shorten >=1pt,auto] (x31) to [] node []{} (d2);
		\draw[-stealth, shorten >=1pt,auto] (nx31) to [] node []{} (d2);
		\draw[-stealth, shorten >=1pt,auto] (d2) to [] node []{} (d2i1);
		\draw[-stealth, shorten >=1pt,auto] (d2) to [] node []{} (d2i2);
		\draw[-stealth, shorten >=1pt,auto] (d2i1) to [] node []{} (dm);
		\draw[-stealth, shorten >=1pt,auto] (d2i2) to [] node []{} (dm);
		\draw[-stealth, shorten >=1pt,auto] (dm) to [] node []{} (x3m);
		\draw[-stealth, shorten >=1pt,auto] (dm) to [] node []{} (nx3m);		
		\draw[-stealth,shorten >=1pt,auto,in=-15,out=15,looseness=4] (x3m) edge [] node {} (x3m);
		\draw[-stealth,shorten >=1pt,auto,in=-15,out=15,looseness=6] (nx3m) edge [] node {} (nx3m);

		\end{tikzpicture}
		}%
	
	\caption{The arena $G$ used in the reduction from the \succinctSetCoverAb{}.}
	\label{sscp_game_long}
\end{figure}

\subparagraph*{Objectives.}
The game is played between Player~$0$ with reachability objective $\ObjPlayer{0}$ and Player~$1$ with $\nbrObjectives = 1 + 2 \cdot m + p + q$ reachability objectives. The payoff of a play therefore consists in a single Boolean for objective $\ObjPlayer{1}$, a vector of $2 \cdot m$ Booleans for objectives $\ObjPlayer{x_1}, \ObjPlayer{\neg x_1}, \dots,  \ObjPlayer{x_m}, \ObjPlayer{\neg x_m}$, a vector of $p$ Booleans for objectives $ \ObjPlayer{C_1}, \dots, \ObjPlayer{C_p}$ and a vector of $q$ Booleans for objectives $\ObjPlayer{D_1}, \dots, \ObjPlayer{D_q}$. The objectives are defined as follows.
\begin{itemize}
    \item The target set for objective $\ObjPlayer{0}$ of Player~$0$ and objective $\ObjPlayer{1}$ of Player~$1$ is $\{v_2, v_3\}$.
    \item The target set for objective $\ObjPlayer{x_i}$ (resp.\ $\ObjPlayer{\neg x_i}$) with $i \in \{1, \dots, m\}$ is the set of vertices labeled $x_i$ (resp.\ $\neg x_i$) in $G_1$, $G_2$ and $G_3$. 
    \item The target set for objective $\ObjPlayer{C_i}$ with $i \in \{1, \dots, p\}$ is the set of vertices in $G_1$ and $G_3$ corresponding to the literals of $X$ which make up the clause $C_i$ in $\phi$. In addition, vertex $i_j$ in $G_2$ belongs to the target set of objective $\ObjPlayer{C_\ell}$ for all $\ell \in \{1, \dots, p\}$ such that $\ell \neq j$. 
    \item The target set of objective $\ObjPlayer{D_i}$ with $i \in \{1, \dots, q\}$ is the set of vertices in $G_3$ corresponding to the literals of $X$ and $Y$ which make up the clause $D_i$ in $\psi$. In addition, vertices $v_1$ and $v_2$ satisfy every objective $\ObjPlayer{D_i}$ with $i \in \{1, \dots, q\}$. 
\end{itemize}

\subparagraph*{Sub-arenas $G_1$ and $G_2$.}
In each sub-arena $G_1$ and $G_2$, for each variable $x_i \in X$, there is one choice vertex controlled by Player~$1$ which leads to $x_i$ and $\neg x_i$. These vertices have the next choice vertex as their successor, except for vertices $x_m$ and $\neg x_m$ which have a self loop. In $G_2$, there is also a vertex $v_2$ controlled by Player~$1$ with $p$ successors, each leading to the first choice vertex for the variables in $X$. Sub-arenas $G_1$ and $G_2$ are completely controlled by Player~$1$. Plays entering these sub-arenas are therefore consistent with any strategy of Player~$0$.

\subparagraph*{Payoff of Plays in $G_1$.}
Plays in $G_1$ do not satisfy objective $\Omega_0$ of Player $0$ nor objective $\ObjPlayer{1}$ of Player~$1$. A play in $G_1$ is of the form $v_0 \: v_1 \: z_1 \boxempty \dots \boxempty (z_m)^\omega$ where $z_i$ is either $x_i$ or $\neg x_i$. It follows that a play satisfies the objective $\ObjPlayer{x_i}$ or $\ObjPlayer{\neg x_i}$ for each $x_i \in X$. The vector of payoffs for these objectives corresponds to a valuation of the variables in $X$, expressed as a vector of $2 \cdot m$ Booleans. In addition, due to the way the objectives are defined, objective $\ObjPlayer{C_i}$ is satisfied in a play if and only if clause $C_i$ of $\phi$ is satisfied by the valuation this play corresponds to. The objective $\ObjPlayer{D_i}$ for $i \in \{1, \dots, q\}$ is satisfied in every play in $G_1$.

\begin{lemma} \label{lem:G1}
    Plays in $G_1$ are consistent with any strategy of Player~$0$. Their payoff are of the form $(0, val, sat(\phi, val), 1, \dots, 1)$ where $val$ is a valuation of the variables in $X$ expressed as a vector of payoffs for objectives $\ObjPlayer{x_1}$ to $\ObjPlayer{\neg x_m}$ and $sat(\phi, val)$ is the vector of payoffs for objectives $\ObjPlayer{C_1}$ to $\ObjPlayer{C_p}$ corresponding to that valuation. All plays in $G_1$ are lost by Player~$0$.
\end{lemma}

\subparagraph*{Payoff of Plays in $G_2$.}
Plays in $G_2$ satisfy the objectives $\ObjPlayer{0}$ of Player~$0$ and $\ObjPlayer{1}$ of Player~$1$. A play in $G_2$ is of the form $v_0 \: v_2 \: i_j \boxempty z_1 \boxempty \dots \boxempty (z_m)^\omega$ where $z_\ell$ is either $x_\ell$ or $\neg x_\ell$. It follows that a play satisfies either the objective $\ObjPlayer{x}$ or $\ObjPlayer{\neg x}$ for each $x \in X$ which again corresponds to a valuation of the variables in $X$. The objective $\ObjPlayer{D_i}$ for $i \in \{1, \dots, q\}$ is satisfied in every play in $G_2$. Compared to the plays in $G_1$, the difference lies in the objectives corresponding to clauses of $\phi$ which are satisfied. In any play in $G_2$, a vertex $i_{j}$ with $j \in \{1, \dots, p\}$ is first visited, satisfying all the objectives $\ObjPlayer{C_\ell}$ with $\ell \in \{1, \dots, p\}$ and $\ell \neq j$. All but one objective corresponding to the clauses of $\phi$ are therefore satisfied. 

\begin{lemma} \label{lem:G2}
    Plays in $G_2$ are consistent with any strategy of Player~$0$. Their payoff are of the form $(1, val, vec, 1, \dots, 1)$ where $val$ is a valuation of the variables in $X$ expressed as a vector of payoffs for objectives $\ObjPlayer{x_1}$ to $\ObjPlayer{\neg x_m}$ and $vec$ is a vector of payoffs for objectives $\ObjPlayer{C_1}$ to $\ObjPlayer{C_p}$ in which all of them except one are satisfied. All plays in $G_2$ are won by Player~$0$.
\end{lemma}

From the two previous lemmas, we can state the following lemma when considering the payoffs of plays in $G_1$ and $G_2$.

\begin{lemma} \label{lem:valsatis}
\label{larger_payoff}
    For every play in $G_1$ which corresponds to a valuation of the variables in $X$ that does not satisfy $\phi$, there is a play in $G_2$ with a strictly larger payoff.
\end{lemma}

\begin{proof}[Proof of \autoref{lem:valsatis}.]
Let $\rho$ be a play in $G_1$ which corresponds to a valuation of the variables in $X$ that does not satisfy $\phi$. It follows that at least one objective, say $\ObjPlayer{C_\ell}$, is not satisfied in $\rho$ as at least one clause of $\phi$ (clause $C_\ell$) is not satisfied by that valuation. Let us consider the play $\rho'$ in $G_2$ which visits vertex $i_\ell$ and after visits the vertices corresponding to the same valuation of the variables in $X$ as $\rho$. By Lemmas~\ref{lem:G1} and~\ref{lem:G2}, it follows that the payoff of $\rho'$ is strictly larger than that of $\rho$ (as we have $(0, val, sat(\phi, val), 1, \dots, 1) < (1, val, vec, 1, \dots, 1)$ with $sat(\phi, val) \leq vec$). 
\end{proof}

The following lemma is a consequence of Lemma~\ref{lem:valsatis}.

\begin{lemma}
\label{g1_g2_payoffs}
	The set of payoffs of plays in $G_1$ that are \paretoOptimal{} when considering $G_1 \cup G_2$ for any strategy $\sigma_0$ of Player~$0$ is equal to the set of payoffs of plays in $G_1$ whose valuation of $X$ satisfy $\phi$.
\end{lemma}

\begin{proof}[Proof of \autoref{g1_g2_payoffs}.]
This property stems from the following observations. First, any play in $G_1$ which satisfies every objective $\ObjPlayer{C_i}$ with $i \in \{1, \dots, p\}$, and therefore corresponds to a valuation of $X$ which satisfies $\phi$, has a payoff that is incomparable to every possible payoff in $G_2$. This is because such a play satisfies more objectives in $\ObjPlayer{C_1}, \dots, \ObjPlayer{C_p}$ than the plays in $G_2$ but does not satisfy objective $\ObjPlayer{1}$ while the plays in $G_2$ do. Second, every other play in $G_1$ has a strictly smaller payoff then at least one play in $G_2$ due to Lemma \ref{larger_payoff} and its payoff is therefore not \paretoOptimal{}.
\end{proof}

\subparagraph*{Problematic Payoffs in $G_1$.}
The plays described in the previous lemma correspond exactly to the valuations of $X$ which satisfy $\phi$ and therefore to the elements we aim to cover in the \succinctSetCoverAb{}. They are \paretoOptimal{} when considering $G_1 \cup G_2$ and are lost by Player~$0$. All other \paretoOptimal{} payoffs in $G_1 \cup G_2$ are only realized by plays in $G_2$ which are all won by Player~$0$. It follows that in order for Player~$0$ to find a strategy $\sigma_0$ from $v_0$ that is solution to the \problemAb{}, it must hold that those payoffs are not \paretoOptimal{} when considering $G_1 \cup G_2 \cup G_3$. Otherwise, a play consistent with $\sigma_0$ with a \paretoOptimal{} payoff is lost by Player~$0$. We call those payoffs \emph{problematic payoffs}.

In order for Player~$0$ to find a strategy $\sigma_0$ which is a solution to the \problemAb{}, this strategy must be such that for each problematic payoff in $G_1$, there is a play in $G_3$ consistent with $\sigma_0$ and with a strictly larger payoff. Since the plays in $G_3$ are all won by Player~$0$, this would ensure that the strategy $\sigma_0$ is a solution to the problem. This corresponds in the \succinctSetCoverAb{} to selecting a series of sets in order to cover the valuations of $X$ which satisfy $\phi$.

\subparagraph*{Sub-arena $G_3$.}
Sub-arena $G_3$ starts with gadget $Q_\problemParam$ whose vertices are controlled by Player~$1$. Then, for each variable $y_i \in Y$, there is one choice vertex controlled by Player~$0$ which leads to $y_i$ and $\neg y_i$. These vertices have the next choice vertex as their successor, except for $y_n$ and $\neg y_n$ which lead to the first choice vertex for the variables in $X$. 

\subparagraph*{Payoff of Plays in $G_3$.}
Plays in $G_3$ satisfy the objectives $\ObjPlayer{0}$ of Player~$0$ and $\ObjPlayer{1}$ of Player~$1$. A play in $G_3$ consistent with a strategy $\sigma_0$ is of the form $v_0 \boxempty \dots \boxempty v_3 \, r_1 \raisebox{0.2ex}{$\scriptstyle\varbigcirc$} \cdots \raisebox{0.2ex}{$\scriptstyle\varbigcirc$} \, r_n \boxempty z_1 \boxempty \dots \boxempty (z_m)^\omega$ where $r_i$ is either $y_i$ or $\neg y_i$ and $z_i$ is either $x_i$ or $\neg x_i$. Since only the vertices leading to $y$ or $\neg y$ for $y \in Y$ belong to Player~$0$, it holds that $v_3 \, r_1 \raisebox{0.2ex}{$\scriptstyle\varbigcirc$} \cdots \raisebox{0.2ex}{$\scriptstyle\varbigcirc$} \, r_n$ is the only part of any play in $G_3$ which is directly influenced by $\sigma_0$. That part of a play comes after a history from $v_0$ to $v_3$ of which there are $\problemParam$, provided by gadget $Q_k$. By definition of a strategy, this can be interpreted as Player~$0$ making a choice of valuation of the variables in $Y$ after each of those $\problemParam$ histories. After this, the play satisfies either the objective $\ObjPlayer{x}$ or $\ObjPlayer{\neg x}$ for each $x \in X$ which corresponds to a valuation of $X$. Due to the way the objectives are defined, the objective $\ObjPlayer{C_i}$ (resp.\ $\ObjPlayer{D_i}$) is satisfied if and only if clause $C_i$ of $\phi$ (resp.\ $D_i$ of $\psi$) is satisfied by the valuation of the variables in $X$ (resp. $X$ and $Y$) the play corresponds to. 

\subparagraph*{Creating Strictly Larger Payoffs in $G_3$.}
In order to create a play with a payoff $r'$ that is strictly larger than a problematic payoff $r$, $\sigma_0$ must choose a valuation of $Y$ such that there exists a valuation of the remaining variables $X$ which together with this valuation of $Y$ satisfies $\psi$ and $\phi$ (since in $r$ every objective $\ObjPlayer{C_i}$ for $i \in \{1, \dots, p\}$ and $\ObjPlayer{D_i}$ for $i \in \{1, \dots, q\}$ is satisfied). Since the plays in $G_3$ also satisfy the objective $\ObjPlayer{1}$ and plays in $G_1$ do not, this ensures that $r < r'$. \\

We can finally establish that our reduction is correct.
\begin{proposition}\label{prop:nexptime-hard-correct-reach}
    Player~$0$ has a strategy $\sigma_0$ from $v_0$ in $G$ that is a solution to the \problemAb{} if and only if there is a solution to the corresponding instance of the \succinctSetCoverAb{}.
\end{proposition}

\begin{proof}[Proof of \autoref{prop:nexptime-hard-correct-reach}.]
Let us assume that that $\sigma_0$ is a solution to the \problemAb{} in $G$ and show that there is a solution to the \succinctSetCoverAb{}. Let $val_X$ be a valuation of the variables in $X$ which satisfies $\phi$. This valuation corresponds to a play in $G_1$ with a problematic payoff $r$. Since the objective of Player~$0$ is not satisfied in that play and since $\sigma_0$ is a solution to the \problemAb{}, it holds that $r$ is not \paretoOptimal{}. It follows that there exists a play in $G_3$ that is consistent with $\sigma_0$ and whose payoff is strictly larger than $r$. As described above, such a play corresponds to a valuation $val_Y$ of the variables in $Y$ such that $val_X \in \llbracket \psi[val_Y] \rrbracket$. Since this can be done for each $val_X \in \llbracket \phi \rrbracket$ and since there is a set $K$ of $\problemParam$ possible valuations $val_Y$ in $G_3$, it holds that $\llbracket \phi \rrbracket \subseteq \bigcup\limits_{val_Y \in K} \llbracket \psi[val_Y] \rrbracket$. 

Let us now assume that there is a solution to the \succinctSetCoverAb{} and show that we can construct a strategy $\sigma_0$ that is solution to the \problemAb{}. Let $K$ be the set of $\problemParam$ valuations $val_Y$ of the variables in $Y$ which is a solution to the \succinctSetCoverAb{}. Since there are $\problemParam$ possible histories from $v_0$ to $v_3$ in $G_3$ provided by the gadget $Q_\problemParam$ described previously, we define $\sigma_0$ such that the $n$ vertices $y_i$ or $\neg y_i$ for $i \in \{1, \dots , n\}$ visited after each history correspond to a valuation in $K$. We can now show that this strategy is a solution to the \problemAb{}. We do this by showing that each play $\rho$ with problematic payoff $r$ in $G_1$ has a strictly smaller payoff than that of some play $\rho'$ with payoff $r'$ in $G_3$. Such a payoff $r$ corresponds to a valuation $val_X \in \llbracket \phi \rrbracket$. Since $K$ is a solution to the \succinctSetCoverAb{}, it holds that there exists some valuation $val_Y \in K$ such that $val_X \in \llbracket \psi[val_Y] \rrbracket$. It follows, given the definition of $\sigma_0$, that there exists a play $\rho'$ in $G_3$ corresponding to that valuation $val_Y$ and which visits the vertices $x$ or $\neg x$ for each $x \in X$ such that it corresponds to the valuation $val_X$. Given the properties mentioned before, the payoff $r'$ of this play is such that $r < r'$.
\end{proof}

The previous proof yields our result on the \nexptime{}-hardness of the \problemAb{} in reachability \gamesAb{}.

\subsubsection{\textsf{NEXPTIME}-Hardness of Parity SP Games}

We now provide the \nexptime-hardness result for parity \gamesAb{}.

The proof of this result follows the same ideas used in the proof for reachability \gamesAb{}. It again uses a reduction from the \succinctSetCoverAb{} in which we construct an arena $G$, and its structure of three sub-arenas $G_1$, $G_2$, and $G_3$ is kept. We describe the main difficulties that we encounter when adapting this proof for parity objectives and how to overcome them by modifying each sub-arena $G_i$ into $G'_i$. The modified arena $G'$ is depicted in Figure~\ref{fig:succinctparity}.

\begin{figure}
	\centering

		\begin{tikzpicture}

		\node[draw, rectangle, minimum size=0.8cm] (a1) at (-2.5,6.25){$v_1$};
		\node[draw, rectangle, minimum size=0.8cm] (x21) at (-1.5,7.25){$x_1$};
		\node[draw, rectangle, minimum size=0.8cm] (nx21) at (-1.5,5.25){$\neg x_1$};
		\node[draw, rectangle, minimum size=0.8cm] (a2) at (-0.5,6.25){};
		\node[minimum size=0.8cm] (a2i1) at (0.5,7.25){$\dots$};
		\node[minimum size=0.8cm] (a2i2) at (0.5,5.25){$\dots$};
		\node[draw, rectangle, minimum size=0.8cm] (am) at (1.5,6.25){};
		\node[draw, rectangle, minimum size=0.8cm] (x2m) at (2.5,7.25){$x_m$};
		\node[draw, rectangle, minimum size=0.8cm] (nx2m) at (2.5,5.25){$\neg x_m$};
		\node[draw, rectangle, minimum size=0.8cm] (d1) at (3.5,6.25){};

		\draw[-stealth, shorten >=1pt,auto] (a1) to [] node []{} (x21);
		\draw[-stealth, shorten >=1pt,auto] (a1) to [] node []{} (nx21);
		\draw[-stealth, shorten >=1pt,auto] (x21) to [] node []{} (a2);
		\draw[-stealth, shorten >=1pt,auto] (nx21) to [] node []{} (a2);
		\draw[-stealth, shorten >=1pt,auto] (a2) to [] node []{} (a2i1);
		\draw[-stealth, shorten >=1pt,auto] (a2) to [] node []{} (a2i2);
		\draw[-stealth, shorten >=1pt,auto] (a2i1) to [] node []{} (am);
		\draw[-stealth, shorten >=1pt,auto] (a2i2) to [] node []{} (am);
		\draw[-stealth, shorten >=1pt,auto] (am) to [] node []{} (x2m);
		\draw[-stealth, shorten >=1pt,auto] (am) to [] node []{} (nx2m);		
		\draw[-stealth, shorten >=1pt,auto] (x2m) to [] node []{} (d1);
		\draw[-stealth, shorten >=1pt,auto] (nx2m) to [] node []{} (d1);
		
		\draw[] (d1) to [] (3.5,7.8125);
		\draw[] (3.5,7.8125) to [] (-2.5,7.8125);
		\draw[-stealth, shorten >=1pt,auto] (-2.5,7.8125) to [] (a1);

		\end{tikzpicture}
	\caption{The repeating structure used in the reduction from the \succinctSetCoverAb{} for parity \gamesAb{}.}
	\label{fig:gadget}
\end{figure}
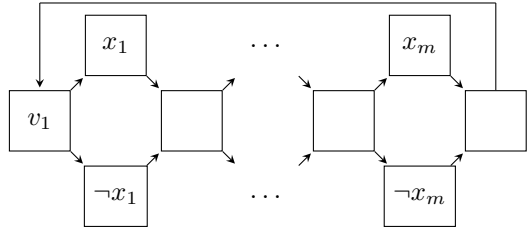

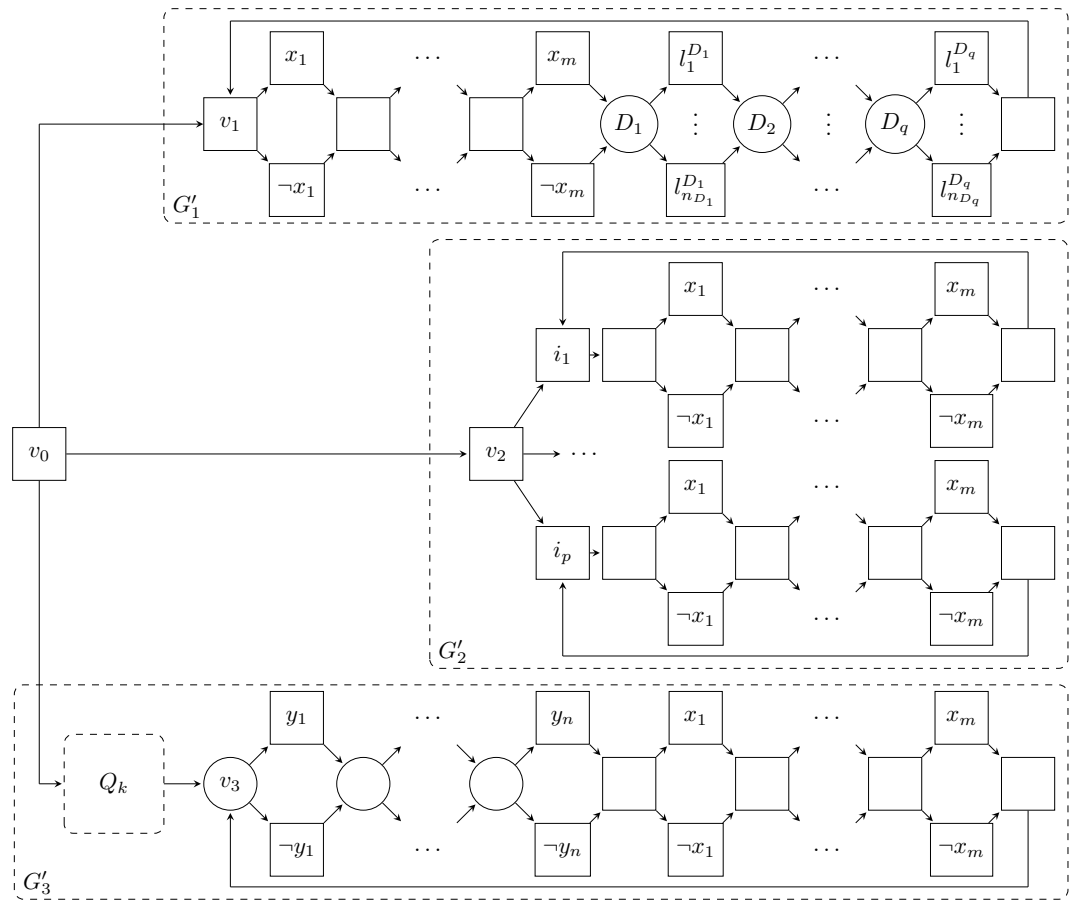
\begin{figure}
	\centering
	\resizebox{\textwidth}{!}{%
		
		\begin{tikzpicture}
		
		\node[draw, rectangle, minimum size=0.8cm] (o) at (-5.375,1.25){$v_0$};
		
		
		\node[draw, rectangle, minimum size=0.8cm] (i0) at (1.5,1.25){$v_2$};
		
		\draw[dashed, rounded corners] (0.5,-2) rectangle (10.1, 4.5)  {};
        \node[] at (0.85,-1.75) {$G'_2$};
		        
		\node[draw, rectangle, minimum size=0.8cm] (i1) at (2.5,2.75){$i_1$};
		\node[] (idot) at (2.85,1.25){$\dots$};
		\node[draw, rectangle, minimum size=0.8cm] (in) at (2.5,-0.25){$i_p$};
		
		\node[draw, rectangle, minimum size=0.8cm] (b1) at (3.5,2.75){};
		\node[draw, rectangle, minimum size=0.8cm] (x11) at (4.5,3.75){$x_1$};
		\node[draw, rectangle, minimum size=0.8cm] (nx11) at (4.5,1.75){$\neg x_1$};
		\node[draw, rectangle, minimum size=0.8cm] (b2) at (5.5,2.75){};
		\node[minimum size=0.8cm] (b2i1) at (6.5,3.75){$\dots$};
		\node[minimum size=0.8cm] (b2i2) at (6.5,1.75){$\dots$};
		\node[draw, rectangle, minimum size=0.8cm] (bm) at (7.5,2.75){};
		\node[draw, rectangle, minimum size=0.8cm] (x1m) at (8.5,3.75){$x_m$};
		\node[draw, rectangle, minimum size=0.8cm] (nx1m) at (8.5,1.75){$\neg x_m$};
		\node[draw, rectangle, minimum size=0.8cm] (end1) at (9.5,2.75){};
		
		\draw[-stealth, shorten >=1pt,auto] (o) to [] node []{} (i0);
		\draw[-stealth, shorten >=1pt,auto] (i0) to [] node []{} (i1);
		\draw[-stealth, shorten >=1pt,auto] (i0) to [] node []{} (2.5,1.25);
		\draw[-stealth, shorten >=1pt,auto] (i0) to [] node []{} (in);
		\draw[-stealth, shorten >=1pt,auto] (i1) to [] node []{} (b1);
		\draw[-stealth, shorten >=1pt,auto] (b1) to [] node []{} (x11);
		\draw[-stealth, shorten >=1pt,auto] (b1) to [] node []{} (nx11);
		\draw[-stealth, shorten >=1pt,auto] (x11) to [] node []{} (b2);
		\draw[-stealth, shorten >=1pt,auto] (nx11) to [] node []{} (b2);
		\draw[-stealth, shorten >=1pt,auto] (b2) to [] node []{} (b2i1);
		\draw[-stealth, shorten >=1pt,auto] (b2) to [] node []{} (b2i2);
		\draw[-stealth, shorten >=1pt,auto] (b2i1) to [] node []{} (bm);
		\draw[-stealth, shorten >=1pt,auto] (b2i2) to [] node []{} (bm);
		\draw[-stealth, shorten >=1pt,auto] (bm) to [] node []{} (x1m);
		\draw[-stealth, shorten >=1pt,auto] (bm) to [] node []{} (nx1m);		
		\draw[-stealth, shorten >=1pt,auto] (x1m) to [] node []{} (end1);
		\draw[-stealth, shorten >=1pt,auto] (nx1m) to [] node []{} (end1);
		
		\draw[] (end1) to [] (9.5, 4.3125);
		\draw[] (9.5, 4.3125) to [] (2.5, 4.3125);
		\draw[-stealth,shorten >=1pt,auto] (2.5, 4.3125) to [] (i1);
		
		\node[draw, rectangle, minimum size=0.8cm] (b1b) at (3.5,-0.25){};
		\node[draw, rectangle, minimum size=0.8cm] (x11b) at (4.5,-1.25){$ \neg x_1$};
		\node[draw, rectangle, minimum size=0.8cm] (nx11b) at (4.5,0.75){$ x_1$};
		\node[draw, rectangle, minimum size=0.8cm] (b2b) at (5.5,-0.25){};
		\node[minimum size=0.8cm] (b2i1b) at (6.5,-1.25){$\dots$};
		\node[minimum size=0.8cm] (b2i2b) at (6.5,0.75){$\dots$};
		\node[draw, rectangle, minimum size=0.8cm] (bmb) at (7.5,-0.25){};
		\node[draw, rectangle, minimum size=0.8cm] (x1mb) at (8.5,-1.25){$\neg x_m$};
		\node[draw, rectangle, minimum size=0.8cm] (nx1mb) at (8.5,0.75){$ x_m$};
		\node[draw, rectangle, minimum size=0.8cm] (end1b) at (9.5,-0.25){};
		
		\draw[-stealth, shorten >=1pt,auto] (in) to [] node []{} (b1b);
		\draw[-stealth, shorten >=1pt,auto] (b1b) to [] node []{} (x11b);
		\draw[-stealth, shorten >=1pt,auto] (b1b) to [] node []{} (nx11b);
		\draw[-stealth, shorten >=1pt,auto] (x11b) to [] node []{} (b2b);
		\draw[-stealth, shorten >=1pt,auto] (nx11b) to [] node []{} (b2b);
		\draw[-stealth, shorten >=1pt,auto] (b2b) to [] node []{} (b2i1b);
		\draw[-stealth, shorten >=1pt,auto] (b2b) to [] node []{} (b2i2b);
		\draw[-stealth, shorten >=1pt,auto] (b2i1b) to [] node []{} (bmb);
		\draw[-stealth, shorten >=1pt,auto] (b2i2b) to [] node []{} (bmb);
		\draw[-stealth, shorten >=1pt,auto] (bmb) to [] node []{} (x1mb);
		\draw[-stealth, shorten >=1pt,auto] (bmb) to [] node []{} (nx1mb);		
		\draw[-stealth, shorten >=1pt,auto] (x1mb) to [] node []{} (end1b);
		\draw[-stealth, shorten >=1pt,auto] (nx1mb) to [] node []{} (end1b);
		
		\draw[] (end1b) to [] (9.5, -1.8125);
		\draw[] (9.5, -1.8125) to [] (2.5, -1.8125);
		\draw[-stealth,shorten >=1pt,auto] (2.5, -1.8125) to [] (in);

		\draw[dashed, rounded corners] (-3.5, 4.75) rectangle (10.1, 8) {};
		\node[] at (-3.15,5) {$G'_1$};
		
		\node[draw, rectangle, minimum size=0.8cm] (a1) at (-2.5,6.25){$v_1$};
		\node[draw, rectangle, minimum size=0.8cm] (x21) at (-1.5,7.25){$x_1$};
		\node[draw, rectangle, minimum size=0.8cm] (nx21) at (-1.5,5.25){$\neg x_1$};
		\node[draw, rectangle, minimum size=0.8cm] (a2) at (-0.5,6.25){};
		\node[minimum size=0.8cm] (a2i1) at (0.5,7.25){$\dots$};
		\node[minimum size=0.8cm] (a2i2) at (0.5,5.25){$\dots$};
		\node[draw, rectangle, minimum size=0.8cm] (am) at (1.5,6.25){};
		\node[draw, rectangle, minimum size=0.8cm] (x2m) at (2.5,7.25){$x_m$};
		\node[draw, rectangle, minimum size=0.8cm] (nx2m) at (2.5,5.25){$\neg x_m$};
		\node[draw, circle, minimum size=0.8cm] (d1) at (3.5,6.25){$D_1$};
		\node[draw, rectangle, minimum size=0.8cm] (d1l3) at (4.5,5.25){$l^{D_1}_{n_{D_1}}$};
		\node[rectangle, minimum size=0.8cm] (d1l2) at (4.5,6.35){$\vdots$};
		\node[draw, rectangle, minimum size=0.8cm] (d1l1) at (4.5,7.25){$l^{D_1}_1$};
		\node[draw, circle, minimum size=0.8cm] (d2) at (5.5,6.25){$D_2$};
		\node[rectangle, minimum size=0.8cm] (d2l3) at (6.5,5.25){$\dots$};
		\node[rectangle, minimum size=0.8cm] (d2l2) at (6.5,6.35){$\vdots$};
		\node[rectangle, minimum size=0.8cm] (d2l1) at (6.5,7.25){$\dots$};
		\node[draw, circle, minimum size=0.8cm] (dq) at (7.5,6.25){$D_q$};
		\node[draw,rectangle, minimum size=0.8cm] (dql3) at (8.5,5.25){$l^{D_q}_{n_{D_q}}$};
		\node[rectangle, minimum size=0.8cm] (dql2) at (8.5,6.35){$\vdots$};
		\node[draw,rectangle, minimum size=0.8cm] (dql1) at (8.5,7.25){$l^{D_q}_1$};
		\node[draw, rectangle, minimum size=0.8cm] (endg1) at (9.5,6.25){};

		\draw[] (o) to [] (-5.375, 6.25);
		\draw[-stealth] (-5.375, 6.25) to [] (a1);
		\draw[-stealth, shorten >=1pt,auto] (a1) to [] node []{} (x21);
		\draw[-stealth, shorten >=1pt,auto] (a1) to [] node []{} (nx21);
		\draw[-stealth, shorten >=1pt,auto] (x21) to [] node []{} (a2);
		\draw[-stealth, shorten >=1pt,auto] (nx21) to [] node []{} (a2);
		\draw[-stealth, shorten >=1pt,auto] (a2) to [] node []{} (a2i1);
		\draw[-stealth, shorten >=1pt,auto] (a2) to [] node []{} (a2i2);
		\draw[-stealth, shorten >=1pt,auto] (a2i1) to [] node []{} (am);
		\draw[-stealth, shorten >=1pt,auto] (a2i2) to [] node []{} (am);
		\draw[-stealth, shorten >=1pt,auto] (am) to [] node []{} (x2m);
		\draw[-stealth, shorten >=1pt,auto] (am) to [] node []{} (nx2m);		
		\draw[-stealth, shorten >=1pt,auto] (x2m) to [] node []{} (d1);
		\draw[-stealth, shorten >=1pt,auto] (nx2m) to [] node []{} (d1);
		\draw[-stealth, shorten >=1pt,auto] (d1) to [] node []{} (d1l1);
		\draw[-stealth, shorten >=1pt,auto] (d1) to [] node []{} (d1l3);
		\draw[-stealth, shorten >=1pt,auto] (d1l1) to [] node []{} (d2);
		\draw[-stealth, shorten >=1pt,auto] (d1l3) to [] node []{} (d2);
		\draw[-stealth, shorten >=1pt,auto] (d2) to [] node []{} (d2l1);
		\draw[-stealth, shorten >=1pt,auto] (d2) to [] node []{} (d2l3);
		\draw[-stealth, shorten >=1pt,auto] (d2l1) to [] node []{} (dq);
		\draw[-stealth, shorten >=1pt,auto] (d2l3) to [] node []{} (dq);		
		\draw[-stealth, shorten >=1pt,auto] (dq) to [] node []{} (dql1);
		\draw[-stealth, shorten >=1pt,auto] (dq) to [] node []{} (dql3);
		\draw[-stealth, shorten >=1pt,auto] (dql1) to [] node []{} (endg1);
		\draw[-stealth, shorten >=1pt,auto] (dql3) to [] node []{} (endg1);
		
		\draw[] (endg1) to [] (9.5,7.8125);
		\draw[] (9.5,7.8125) to [] (-2.5,7.8125);
		\draw[-stealth, shorten >=1pt,auto] (-2.5,7.8125) to [] (a1);

		
		\draw[dashed, rounded corners] (-5.75,-5.5) rectangle (10.1, -2.25) {};
		\node[] at (-5.4,-5.25) {$G'_3$};

		\draw[dashed, rounded corners] (-5,-4.5) rectangle (-3.5, -3) {};
		\node[rectangle, minimum size=0.8cm] (j1) at (-4.25,-3.75){$Q_\problemParam$};

		\draw[] (o) to [] ((-5.375,-3.75);
		\draw[-stealth, shorten >=1pt,auto] (-5.375,-3.75) to [] ((-5,-3.75);

		
		\node[draw, circle, minimum size=0.8cm] (c1) at (-2.5,-3.75){$v_3$};
		\node[draw, rectangle, minimum size=0.8cm] (y1) at (-1.5,-4.75){$\neg y_1$};
		\node[draw, rectangle, minimum size=0.8cm] (ny1) at (-1.5,-2.75){$y_1$};
		\node[draw, circle, minimum size=0.8cm] (c2) at (-0.5,-3.75){};
		\node[minimum size=0.8cm] (ci1) at (0.5,-4.75){$\dots$};
		\node[minimum size=0.8cm] (ci2) at (0.5,-2.75){$\dots$};
		\node[draw, circle, minimum size=0.8cm] (cm) at (1.5,-3.75){};
		\node[draw, rectangle, minimum size=0.8cm] (ym) at (2.5,-4.75){$\neg y_n$};
		\node[draw, rectangle, minimum size=0.8cm] (nym) at (2.5,-2.75){$y_n$};
		
		\draw[-stealth, shorten >=1pt,auto] ((-3.5,-3.75) to [] node []{} (c1);

		\draw[-stealth, shorten >=1pt,auto] (c1) to [] node []{} (y1);
		\draw[-stealth, shorten >=1pt,auto] (c1) to [] node []{} (ny1);
		\draw[-stealth, shorten >=1pt,auto] (y1) to [] node []{} (c2);
		\draw[-stealth, shorten >=1pt,auto] (ny1) to [] node []{} (c2);
		\draw[-stealth, shorten >=1pt,auto] (c2) to [] node []{} (ci1);
		\draw[-stealth, shorten >=1pt,auto] (c2) to [] node []{} (ci2);
		\draw[-stealth, shorten >=1pt,auto] (ci1) to [] node []{} (cm);
		\draw[-stealth, shorten >=1pt,auto] (ci2) to [] node []{} (cm);
		\draw[-stealth, shorten >=1pt,auto] (cm) to [] node []{} (ym);
		\draw[-stealth, shorten >=1pt,auto] (cm) to [] node []{} (nym);			
		
		
		\node[draw, rectangle, minimum size=0.8cm] (d1) at (3.5,-3.75){};
		\node[draw, rectangle, minimum size=0.8cm] (x31) at (4.5,-4.75){$\neg x_1$};
		\node[draw, rectangle, minimum size=0.8cm] (nx31) at (4.5,-2.75){$x_1$};
		\node[draw, rectangle, minimum size=0.8cm] (d2) at (5.5,-3.75){};
		\node[minimum size=0.8cm] (d2i1) at (6.5,-4.75){$\dots$};
		\node[minimum size=0.8cm] (d2i2) at (6.5,-2.75){$\dots$};
		\node[draw, rectangle, minimum size=0.8cm] (dm) at (7.5,-3.75){};
		\node[draw, rectangle, minimum size=0.8cm] (x3m) at (8.5,-4.75){$\neg x_m$};
		\node[draw, rectangle, minimum size=0.8cm] (nx3m) at (8.5,-2.75){$x_m$};
		\node[draw, rectangle, minimum size=0.8cm] (eng3) at (9.5,-3.75){};
				
		\draw[-stealth, shorten >=1pt,auto] (ym) to [] node []{} (d1);		
		\draw[-stealth, shorten >=1pt,auto] (nym) to [] node []{} (d1);	
		\draw[-stealth, shorten >=1pt,auto] (d1) to [] node []{} (x31);
		\draw[-stealth, shorten >=1pt,auto] (d1) to [] node []{} (nx31);
		\draw[-stealth, shorten >=1pt,auto] (x31) to [] node []{} (d2);
		\draw[-stealth, shorten >=1pt,auto] (nx31) to [] node []{} (d2);
		\draw[-stealth, shorten >=1pt,auto] (d2) to [] node []{} (d2i1);
		\draw[-stealth, shorten >=1pt,auto] (d2) to [] node []{} (d2i2);
		\draw[-stealth, shorten >=1pt,auto] (d2i1) to [] node []{} (dm);
		\draw[-stealth, shorten >=1pt,auto] (d2i2) to [] node []{} (dm);
		\draw[-stealth, shorten >=1pt,auto] (dm) to [] node []{} (x3m);
		\draw[-stealth, shorten >=1pt,auto] (dm) to [] node []{} (nx3m);
		\draw[-stealth, shorten >=1pt,auto] (x3m) to [] node []{} (eng3);
		\draw[-stealth, shorten >=1pt,auto] (nx3m) to [] node []{} (eng3);

		\draw[] (eng3) to [] (9.5,-5.3125);
		\draw[] (9.5,-5.3125) to [] (-2.5,-5.3125);
		\draw[-stealth,shorten >=1pt,auto] (-2.5,-5.3125) to [] (c1);
		
		\end{tikzpicture}
		}%
	\caption{The arena $G'$ used in the reduction from the \succinctSetCoverAb{} for parity \gamesAb{}.}
	\label{fig:succinctparity}
\end{figure}

\subparagraph*{Two Particular Objectives.}
Remember that for the case of reachability \gamesAb{}, the objective $\Omega_0$ of Player~$0$ and the first objective $\Omega_1$ of Player~$1$ were either always satisfied or always not satisfied in all plays of a given sub-arena $G_1$, $G_2$, or $G_3$. This property holds in the modified arena $G'$, with these objectives being expressed using parity conditions. 

\subparagraph*{On the Choice of Valuations.}
Recall that for reachability objectives, we encoded the choice of valuations for the variables $x_i \in X$ made by Player~$1$ by using sub-arena $G_1$ of Figure~\ref{sscp_game_long} (which was also reused as part of $G_2$ and $G_3$). A play in this sub-arena encodes a valuation by visiting one literal $l_i \in \{x_i, \neg x_i\}$ for each $x_i \in X$ thus satisfying the objectives $\Omega_{l_i}$, $i \in \{1,\ldots,m\}$. 

This simple schema cannot be reused in the case of parity objectives as they are prefix-independent objectives. Instead, we ask Player~$1$ to repeatedly produce the \emph{same choice} along loops (from $v_1$ back to $v_1$) in the adapted gadget of Figure~\ref{fig:gadget}.

We associate with each variable $x_i \in X$ two parity objectives with priority function $c_{l_i}$ with $l_i \in \{x_i, \neg x_i\}$ defined as follows: $c_{l_i}(l_i) = 2$, $c_{l_i}(\neg l_i) = 1$ and $c_{l_i}(v) = 3$ for all the other vertices. It is easy to see that plays in which the valuation changes infinitely many times (for example, visiting $x_i$ and $\neg x_i$ infinitely often), have a payoff strictly smaller than some other play which settles on a choice of valuation for each variable. Indeed, the payoff for objectives $\Omega_{x^1_i}$ and $\Omega_{x^2_i}$ is $(0,0)$ in the first case and $(1,0)$ or $(0,1)$ in the second. We say that a valuation is \emph{properly encoded} if Player~$1$ eventually repeats the same choice along the loops to settle on a valuation.

Notice that the gadget of Figure~\ref{fig:gadget} appears nearly identical as part of the three sub-arenas $G'_1$, $G'_2$, and $G'_3$ of Figure~\ref{fig:succinctparity}. In all of these sub-arenas, we define the values of each priority function $c_{l_i}$ with $l_i \in \{x_i, \neg x_i\}$ exactly as explained above.

\subparagraph*{On the Satisfied Clauses.}
Recall that in the case of reachability \gamesAb{}, we associated one reachability objective with each clause $C_i$ (resp. $D_i$) of $\phi$ (resp. $\psi$). As encoding valuations is made more complex by the prefix-independency of parity objectives, we also need to adapt the way we check which clauses are satisfied by a given valuation. 

Let us first explain our encoding for clauses $C_j$ of $\phi$. We associate one parity objective with priority function $c_l^{C_j}$ with \emph{each literal} $l$ of each clause $C_j$. Therefore, if $C_j = l_1^{C_j} \vee \ldots \vee l_{n_{C_j}}^{C_j}$, there are $n_{C_j}$ parity objectives for clause $C_j$. Priority function $c_l^{C_j}$ is defined as follows for each vertex $l_i \in \{x_i,\neg x_i\}$ that appears in $G'_1$ and $G'_3$ (we will define it later for $G'_2$): $c_{l}^{C_j}(l_i) = 2$ if $l = l_i$ and $c_{l}^{C_j}(l_i) = 1$ if $l = \neg l_i$, and for all the other vertices $v$, we have $c_{l}^{C_j}(v) = 3$. This encoding of clauses has the following important property: given a valuation $val_X$ of the variables in $X$ properly encoded by Player~$1$, a clause $C_j$ is satisfied by $val_X$ if and only if the parity condition $\parity{c_l^{C_j}}$ is satisfied for \emph{at least one} of the literals $l$ of $C_j$. Thus with the proposed encoding with priority functions, there are several ways to observe that a clause $C_j$ is satisfied (the corresponding payoff is a \emph{non-null} vector of $n_{C_j}$ Booleans).

In the sub-arena $G'_3$, we see a part resembling the gadget of Figure~\ref{fig:gadget}, however made of vertices $y_i$, $i \in \{1,\ldots,n\}$. This part is related to the choice of a valuation of the variables in $Y$ made by Player~$0$ with respect to $\psi$. To encode clauses $D_j$ of $\psi$, we proceed exactly as we did previously with clauses $C_j$ of $\phi$. We associate one priority function $c_l^{D_j}$ with each literal $l$ of each clause $D_j$ (recall that such a literal uses both sets of variables $X$ and $Y$). We similarly define the values of $c_l^{D_j}$ for vertices of $G'_3$: for $l' \in \{x_i, \neg x_i \mid i \in \{1,\ldots,m\}\} \cup \{y_i,\neg y_i\ \mid i \in \{1,\ldots,n\}\}$, we define $c_{l}^{D_j}(l') = 2$ if $l = l'$ and $c_{l}^{D_j}(l') = 1$ if $l = \neg l'$, and for all the other vertices $v$ of $G'_3$, we define $c_{l}^{D_j}(v) = 3$. Notice that the definition of $c_l^{D_j}$ is given for $G'_3$ only. We will later give its definition for $G'_1$ and $G'_2$.



\subparagraph*{Modifications Needed on $G_1$.}
Let us now explain how to modify $G_1$ into $G'_1$. Remember that in the case of reachability \gamesAb{}, the objective associated with each clause $D_j$ of $\psi$ is satisfied by all plays in $G_1$. Indeed the purpose of $G_1$ (in combination with $G_2$) is to isolate all the encodings of the valuations of $X$ that satisfy $\phi$ independently of $\psi$. In case of a positive instance of the \succinctSetCoverAb{}, the payoff of these encodings are strictly smaller than that of some play in $G_3$, which have to satisfy all clauses of $\psi$ by definition of this problem.

We proceed similarly in $G'_1$. However as there are several ways to satisfy $D_j$ in $\psi$ (at least one of its literals has to be satisfied), we let Player~$0$ choose which way to do it. This is encoded by the part of $G'_1$ made with vertices $D_j$, $j \in \{1, \ldots ,q\}$, controlled by Player~$0$, and their successors $l_1^{D_j}, \ldots, l_{n_{D_j}}^{D_j}$ such that $D_j = l_1^{D_j} \vee \ldots \vee l_{n_{D_j}}^{D_j}$. Given a properly encoded $X$ valuation $val_X$ made by Player~$1$, Player~$0$ chooses a $Y$ valuation $val_Y$ such that if $val_X \in \llbracket \phi \rrbracket$, then $val_X \in \llbracket \psi[val_Y] \rrbracket$. Player~$0$ makes such a choice by selecting at least one literal $l^{D_j}$ of $D_j$, for each clause $D_j$ of $\psi$, such that $l$ is satisfied by the valuation made of $val_X$ and $val_Y$. This is encoded in $G'_1$ by defining the priority function $c_l^{D_j}$ such that $c_l^{D_j}(l) = 2$ and $c_l^{D_j}(v) = 3$ for all the other vertices $v$ of $G'_1$.

\subparagraph*{Modifications Needed on $G_2$.}
Recall that in the case of reachability \gamesAb{}, the construction of $G_2$ ensures that the plays of $G_1$ whose payoff is not strictly smaller than the payoff of a play in $G_2$ are exactly the plays of $G_1$ that encode $X$ valuations $val_X$ such that $val_X \in \llbracket \phi \rrbracket$. As a consequence, the objectives that are satisfied by the plays in $G_2$ are exactly \emph{(i)} those associated with a valuation $val_X$, \emph{(ii)} the objectives associated with all clauses $D_j$ of $\psi$, and \emph{(iii)} the objectives associated with all clauses $C_j$ of $\phi$ except one.

We achieve the same requirement for parity \gamesAb{} by using $G'_2$ in place of $G_2$. In this sub-arena, Player~$1$ first selects one clause $C$ in $\phi$ and then in the selecting part of $G'_2$, the priority functions are defined as follows.
\begin{itemize}
    \item We use the priority functions $c_{l_i}$ with $l_i \in \{x_i, \neg x_i \mid i \in \{1, \ldots, m \}\}$ as defined above for encoding $X$ valuations. 
    \item The priority functions $c_l^{D_j}$ are all defined such that the associated objective $\parity{c_l^{D_j}}$ is satisfied.
    \item Similarly the priority functions $c_l^{C_j}$ are defined such that the associated objective $\parity{c_l^{C_j}}$ is satisfied, except for all priority functions $c_l^{C_j}$ such that $C_j = C$ for which this objective is not satisfied.
\end{itemize}

\subparagraph*{Modifications Needed on $G_3$.}

We have modified $G_1$ into $G'_1$ and $G_2$ into $G'_2$ such that the only plays in $G'_1$ with a Pareto-optimal payoff when considering $G'_1 \cup G'_2$, given any strategy of Player~$0$ are those that encode valuations $val_X \in \llbracket \phi \rrbracket$. In $G'_1$, after such a valuation chosen by Player~$1$, Player~$0$ indicates which valuation $val_Y$ to use such that $val_X \in \llbracket \psi[val_Y] \rrbracket$. He chooses this valuation $val_Y$ by indicating for each clause $D_j$ which literals of $D_j$ he has chosen such that the valuation  
made of $val_X$ and $val_Y$ satisfies $D_j$.

Let us now explain how to modify $G_3$ into $G'_3$. After each of the $k$ histories produced by gadget $Q_k$, both players have to choose some valuation (resp. $val_X$ and $val_Y$) for the variables that they control. In case of a positive instance of \succinctSetCoverAb{}, Player~$0$ will be able to select one of the valuations $val_Y$ that he used in $G'_1$ such that $val_X \in \llbracket \psi[val_Y] \rrbracket$ whenever $val_X \in \llbracket \phi \rrbracket$. It follows that plays in $G'_3$ have a larger payoff than the Pareto-optimal payoffs of $G'_1 \cup G'_2$. In case of a negative instance of \succinctSetCoverAb{}, Player~$0$ will not be able to do so.

Clearly there exists a solution to the \problemAb{} in the modified arena $G'$ if and only if the instance of the \succinctSetCoverAb{} is positive. 

\section{Conclusion}
\label{sec:conclusion}
We have introduced in this paper the class of two-player \gamesAb{} and the \problemAb{} in those games. We provided a reduction from \gamesAb{} to a two-player zero-sum game called the \challengerProverAb{} game in order to provide \FPT{} results on solving this problem. We then showed how the arena and the generic objective of this \challengerProverAb{} game can be adapted to specifically handle reachability and parity \gamesAb{}. This allowed us to prove that reachability (resp. parity) \gamesAb{} are in \FPT{} when the number $\nbrObjectives$ of objectives of Player~$1$ (resp. when $\nbrObjectives$ and the maximal priority according to each priority function in the game) is a parameter. We then turned to the complexity class of the \problemAb{} and provided a proof of its \nexptime-membership, which relied on showing that any solution to the \problemAb{} in a reachability or parity \gameAb{} can be transformed into a solution with an exponential memory. We provided a proof of the \np-completeness of the problem in the simple setting of reachability \gamesAb{} played on tree arenas. We then came back to regular game arenas and provided the proof of the \nexptime-hardness of the \problemAb{} in reachability and parity \gamesAb{}. This proof relied on a reduction from the \succinctSetCoverAb{} which we proved to be \nexptimeComplete{}, a result of potential independent interest. 

In future work, we want to study other $\omega$-regular objectives as well as quantitative objectives such as mean-payoff in the framework of \gamesAb{} and the \problemAb{}. It would also be interesting to study whether other works, such as rational synthesis, could benefit from the approaches used in this paper.



\bibliography{bibliography}

\appendix

\section{Useful Result on SP Games}
\label{app:usefull_notions}

\begin{proposition}
\label{prop:transform_two_suc}
Every parity (resp. reachability) \gameAb{} $\mathcal{G}$ with arena $G$ containing $n$ vertices can be transformed into a parity (resp. reachability) \gameAb{} $\bar{\mathcal{G}}$ with arena $\bar{G}$ containing at most $n^2$ vertices such that any vertex in $\bar{G}$ has at most $2$ successors and Player~$0$ has a strategy $\sigma_0$ that is solution to the \problemAb{} in $G$ if and only if Player~$0$ has a strategy $\bar{\sigma}_0$ that is solution to the problem in $\bar{G}$.
\end{proposition}

\begin{proof}[Proof of \autoref{prop:transform_two_suc}.]
Let $\mathcal{G}$ be an \gameAb{} with arena $G$. Let us first describe the arena $\bar{G}$ of $\bar{\mathcal{G}}$. Let $v \in V$ be a vertex of $G$, then $v$ is also a vertex of $\bar{G}$ such that it belongs to the same player and is the root of a complete binary tree with $\ell = |\{v' \mid (v, v') \in E\}|$ leaves if $(v, v) \not \in E$. Otherwise, $v$ has a self loop and its other successor is the root of such a tree with $\ell-1$ leaves. The internal vertices of the tree (that is vertices which are not $v$, nor the leaves) belong to the same player as $v$. Each leaf vertex $v'$ of this tree is such that $(v, v') \in E$, belongs to the same player as in $G$ and is again the root of its own tree. The initial vertex $v_0$ of $G$ remains unchanged in $\bar{G}$. Since every vertex in $\bar{G}$ is part of a binary tree or has a self loop and a single successor, it holds that it has at most two successors. Since $G$ is a game arena, this transformation is such that each vertex in $\bar{G}$ has at least one successor. It follows that $\bar{G}$ is a game arena containing $n$ vertices $v \in V$ and at most $n - 1$ internal vertices per tree in the case where $v \in V$ has $n$ successors in $G$. It follows that the number of vertices in $\bar{G}$ is at most $n + n \cdot (n - 1) = n^2$. If $\mathcal{G}$ is a reachability \gameAb{}, the target sets remain unchanged in $\bar{G}$. For parity \gamesAb{}, the priority function $c$ remains unchanged for vertices $v \in V$ and we define $c(v') = c(v)$ for $v' \in \bar{V}\setminus V$ such that $v'$ is an internal vertex of a tree whose root is $v$. 

Let us now show that there is a solution to the \problemAb{} in $\mathcal{G}$ if and only if there is a solution in $\bar{\mathcal{G}}$. From each root $v$ of a tree in $\bar{G}$ (corresponding to a vertex $v$ of Player~$i$ in $G$) there is a set of $\ell = |\{v' \mid (v, v') \in E\}|$ different paths controlled by  Player~$i$, each leading to a vertex $v'$. It follows that there exists a play $\rho = v_0 v_1 v_2 \ldots \in \Plays_G$ if and only if there exists a play $\rho' = v_0 a_0 \dots a_{n_1} v_1 b_0 \dots b_{n_2} v_2 \ldots \in \Plays_{\bar{G}}$ such that every vertex $a_i$ (resp.\ $b_i$) belongs to the same player as $v_0$ (resp.\ $v_1$) and so on. Given the way the objectives are defined, in the case of reachability or parity \gamesAb{}, it holds that $\payoff{\rho} = \payoff{\rho'}$ and $\won{\rho} = \won{\rho'}$. Therefore, a strategy $\sigma_0$ that is solution to the \problemAb{} in $G$ can be transformed into a strategy $\bar{\sigma}_0$ which is a solution in $\bar{G}$ and vice-versa.
\end{proof}

\end{document}